\documentclass[12pt, reqno]{amsart}
\usepackage[margin=1in]{geometry}
\usepackage{graphicx}
\usepackage{amssymb}
\usepackage{color}
\usepackage{hyperref}
\usepackage{tikz}
\usetikzlibrary{decorations.markings}
\usepackage{mathtools}
\usepackage[show]{ed}
\usepackage{comment}
\usepackage{enumitem}
\usepackage[normalem]{ulem}

\newtheorem{theorem}{Theorem}

\newtheorem{lemma}[theorem]{Lemma}

\newtheorem*{conjecture}{Conjecture}

\theoremstyle{definition}
\newtheorem{definition}[theorem]{Definition}
\newtheorem{remark}[theorem]{Remark}

\numberwithin{equation}{section}
\numberwithin{theorem}{section}

\newcommand{\p}{\partial}

\newcommand{\bbC}{\mathbb C}
\newcommand{\RR}{\mathbb R}
\newcommand{\ZZ}{\mathbb Z}
\newcommand{\TT}{\mathbb{T}}

\renewcommand{\Re}{\operatorname{Re}}
\renewcommand{\Im}{\operatorname{Im}}

\newcommand{\CornerP}{\mathcal C}
\newcommand{\ssQ}{\mathcal{Q}}  

\newcommand{\Eb}{\mathbf{E}}
\newcommand{\eb}{\mathbf{e}}
\newcommand{\cc}[1]{\overline{#1}}

\DeclareMathOperator{\Hess}{Hess}
\DeclareMathOperator{\tr}{tr}
\DeclareMathOperator{\rk}{rk}
\DeclareMathOperator{\In}{In}

\DeclareMathOperator{\diag}{diag}
\DeclareMathOperator{\Ran}{Ran}
\DeclareMathOperator{\Null}{Null}
\DeclareMathOperator{\col}{col}

\title[A local test for global extrema in the dispersion
relation]{A local test for global extrema in the dispersion
 relation of a periodic graph}

\author{G. Berkolaiko}
\address{Department of
  Mathematics, Texas A\&M University, College Station, TX 77843-3368, USA}
\email{berko@math.tamu.edu}

\author{Y. Canzani}
\address{Department of Mathematics, University of North Carolina at Chapel Hill,
Phillips Hall, Chapel Hill, NC  27599, USA}
\email{canzani@email.unc.edu}

\author{G. Cox}
\address{Department of Mathematics and Statistics, Memorial University of Newfoundland, St. John's, NL A1C 5S7, Canada}
\email{gcox@mun.ca}

\author{J.L. Marzuola}
\address{Department of Mathematics, University of North Carolina at Chapel Hill,
Phillips Hall, Chapel Hill, NC  27599, USA}
\email{marzuola@math.unc.edu}

\begin{document}

\begin{abstract}
  We consider a family of periodic tight-binding models (combinatorial
  graphs) that have the minimal number of links between copies of the
  fundamental domain.  For this family we establish a local condition
  of second derivative type under which the critical points of the
  dispersion relation can be recognized as global maxima or minima.
  Under the additional assumption of time-reversal symmetry, we show
  that any local extremum of a dispersion band is in fact a global
  extremum if the dimension of the periodicity group is three or less, or
  (in any dimension) if the critical point in question is a symmetry
  point of the Floquet--Bloch family with respect to complex
  conjugation.  We demonstrate that our results are nearly optimal
  with a number of examples.
\end{abstract}

\maketitle

\section{Introduction}

Wave propagation through periodic media is usually studied using the
Floquet--Bloch transform (\cite{AshcroftMermin_solid,Kuc_bams16}),
which reduces a periodic eigenvalue problem over an infinite domain to
a parametric family of eigenvalue problems over a compact domain.  In
the tight-binding approximation often used in physical applications,
the wave dynamics is described mathematically in terms of a periodic
self-adjoint operator $H$ acting on $\ell^2(\Gamma)$, where $\Gamma$
is a $\mathbb{Z}^d$-periodic graph (see examples in
Figure~\ref{fig:more_lattices}) and $d$ is the dimension of the
underlying space.  The Floquet--Bloch transform introduces $d$
parameters $\alpha=(\alpha_1,\ldots,\alpha_d)$, called
\emph{quasimomenta}, which take their values in the torus
$\TT^d := \RR^d / (2\pi\ZZ)^d$, called the \emph{Brillouin zone}.  The
transformed operator $T(\alpha)$ is an $N\times N$ Hermitian matrix
function that depends smoothly on $\alpha$; here $N$ is the number of
vertices in a fundamental domain for $\Gamma$.  The graph of the
eigenvalues of $T(\alpha)$, when thought of as a multi-valued function
of $\alpha$, is called the \emph{dispersion relation}.  Indexing the
eigenvalues in increasing order, we refer to the graph of the
$n$-th eigenvalue, $\lambda_n(\cdot)$, as the \emph{$n$-th branch} of
the dispersion relation.  The range of $\lambda_n(\cdot)$ is called
the \emph{$n$-th spectral band}. The union of the spectral bands is
the spectrum of the periodic operator $H$ on $\ell^2(\Gamma)$, the set
of wave energies at which waves can propagate through the medium.
The band edges mark the boundary\footnote{Assuming the bands do not
  overlap; if the edges for each band are found, this can be easily
  verified.}  between propagation and insulation, and are thus of
central importance to understanding physical properties of the
periodic material, see
\cite{AshcroftMermin_solid,ozawa2019topological,Kol_etal20} and
references therein.
  
Naturally, the upper (or lower) edge of the $n$-th band is the maximum (or minimum) value of $\lambda_n(\cdot)$.  Since searching
for the location of the band edges over the whole torus $\TT^d$ can be
computationally intensive, the usual approach is to check several
points of symmetry and lines between them.  However, as shown in
\cite{HarKucSob_jpa07}, extrema of the dispersion relation in $d>1$ do
not have to occur at the symmetry points.
Remarkably, in the present work we show that this problem can be
overcome on graphs that have ``one crossing edge per generator,'' a
property which we now define. A notable example of a graph with this
property is the graph found in \cite{HarKucSob_jpa07} and shown here
in Figure~\ref{fig:periodic_example}.

\begin{figure}
  \centering
  \includegraphics{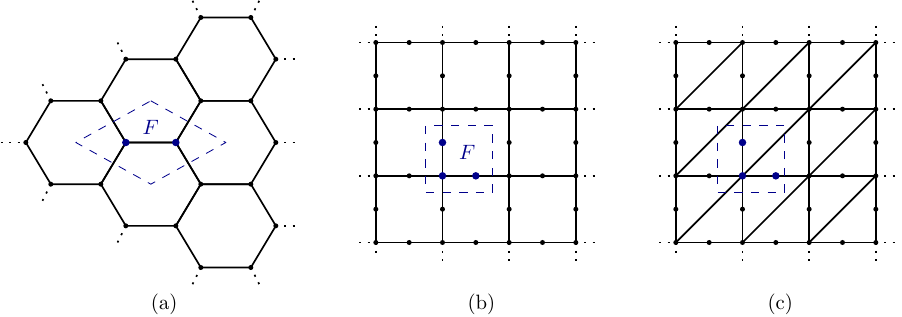}
  \caption{The honeycomb lattice (a) and the Lieb lattice (b) satisfy
    Definition~\ref{def:crossing_edge}, while the augmented Lieb
    lattice (c) does not.  In all figures, the vertices within the
    dashed line show a possible choice of the fundamental domain $F$.}
  \label{fig:more_lattices}
\end{figure}

\begin{definition}
  \label{def:crossing_edge}
  Let  $\Gamma = (V,\sim)$ be a $\mathbb{Z}^d$-periodic graph (see Definition \ref{def:periodic_graph}),  
  where $V$ denotes the set of vertices and $\sim$ denotes the adjacency relation.
  $\Gamma$ is said to have \emph{one crossing edge per
    generator} if it is connected and there exists a choice of a fundamental domain
  $F$ such that there are exactly $2d$ adjacent pairs $u\sim w$ with
  $u \in F$ and $w \in V\setminus F$.
\end{definition}
By a fundamental domain $F$ we mean a subset of $V$ containing exactly
one representative from each orbit generated by the group action of
$\mathbb{Z}^d$.  The choice of a fundamental domain is clearly
non-unique.  In terms of the operator $H$, the edges (adjacency)
denote the interacting pairs of vertices, see \eqref{eq:Hdef} for
details.  We are thus talking about the models known in physics as
``nearest neighbor tight binding''; we stress, however, that our
periodic graphs have arbitrary structure modulo the assumption of
Definition~\ref{def:crossing_edge}.

To give some examples, the ``one crossing edge per generator" 
assumption is satisfied by the $\mathbb{Z}^d$ lattice,
the honeycomb lattice shown in Figure~\ref{fig:more_lattices}(a), and the Lieb
lattice in Figure~\ref{fig:more_lattices}(b). The graph shown in 
Figure~\ref{fig:more_lattices}(c)  does not satisfy
Definition~\ref{def:crossing_edge}.  For further insight into Definition \ref{def:crossing_edge}  see the discussion around equation \eqref{eq:nn_fund}, and  see Figure \ref{fig:periodic_example} for another example.

In this work we prove that for graphs with one crossing edge per
generator, there is a simple \emph{local} criterion\,---\,a variation of the
second derivative test\,---\,that detects \emph{if a given critical
  point of $\lambda_n(\cdot)$ is a global extremum}.  In many cases 
  we can conclude that \emph{any local extremum of a band of the
  dispersion relation is in fact a global extremum}.  This does not
imply uniqueness of, say, a local minimum, but it does mean that every
local minimum attains the same value; see, for example,
Figure~\ref{fig:KucMag}(left).  In a sense, the dispersion relation
behaves as if it were a convex function (even though this can never be
the case for a continuous function on a torus).  As a consequence,
even if no local extrema are found among the points of symmetry, it
would be enough to run a gradient search-like method.

We now formally state our results. For each $1\leq n \leq N$, we are interested in the extrema of the continuous function
$$\alpha \mapsto \lambda(\alpha) := \lambda_n\big(T(\alpha)\big).$$
Assuming the eigenvalue is simple\footnote{If the eigenvalue is
  multiple, then two or more branches touch.  This situation
  is important in applications; there are fast algorithms to find such
  points \cite{DiePug_siamjmaa09,DiePapPug_siamjmaa13,BerPar_siamjmaa21}
  which lie outside the scope of this work.}  at a point
$\alpha^\circ$, $\lambda(\alpha)$ is a real analytic function of
$\alpha$ in a neighborhood of $\alpha^\circ$, by \cite[Section II.6.4]{K76}.

To look for the critical points of $\lambda(\alpha)$ and to test their
\emph{local} character, one can use the following formulas (see
Section~\ref{sec:var flas}) for the first two derivatives of a simple
eigenvalue $\lambda(\alpha)$:
\begin{equation}
  \label{eq:derivatives_lambda}
  \nabla \lambda(\alpha^\circ)
  = B^* f^\circ,
  \qquad \qquad
  \Hess \lambda(\alpha^\circ) = 2 \Re W,
\end{equation}
where
\begin{equation}
  \label{eq:W_def}
  W := \Omega - B^* \big(T(\alpha^\circ) - \lambda(\alpha^\circ)\big)^+ B,
\end{equation}
$f^\circ$ is the normalized eigenvector corresponding to the
eigenvalue $\lambda(\alpha^\circ)$ of $T(\alpha^\circ)$, $B$ and
$\Omega$ are correspondingly the
$N\times d$ matrix of first derivatives and
$d\times d$ matrix of second derivatives of $T(\alpha)$ at
$\alpha=\alpha^\circ$ evaluated on $f^\circ$,
\begin{align}
  \label{eqn:BandOmega_def} 
  B  := D \big(T(\alpha) f^\circ \big)\Big|_{\alpha=\alpha^\circ},
  \qquad
  \Omega :=
  \frac12 \Hess \left\langle f^\circ, T(\alpha) f^\circ \right\rangle
  \Big|_{\alpha = \alpha^\circ},
\end{align}
and $\big(T(\alpha^\circ) - \lambda(\alpha^\circ)\big)^+$ denotes the
Moore--Penrose pseudoinverse of
$T(\alpha^\circ) - \lambda(\alpha^\circ)$.

The textbook second derivative test tells us that a point
$\alpha^\circ$ with $B^* f^\circ=0$ and $\Re W > 0$ is a local
minimum.  It turns out that a lot more information can be gleaned from
the matrix $W$ itself, which may be complex.

\begin{theorem}
  \label{thm:mainI}
  Let $\Gamma$ be a $\mathbb{Z}^d$-periodic graph with one crossing 
  edge per generator, and let $H$ be a periodic self-adjoint operator 
  acting on $\ell^2(\Gamma)$.
  Suppose that the $n$-th branch,
  $\lambda(\alpha) = \lambda_n\big(T(\alpha)\big)$, of the Floquet--Bloch
  transformed operator $T(\alpha)$ has a critical point at
  $\alpha^\circ \in \TT^d$.
  Suppose that $\lambda(\alpha^\circ)$ is a simple eigenvalue of $T(\alpha^\circ)$
  and that the corresponding eigenvector $f^\circ$ is non-zero on at least
  one end of any crossing edge.  Let $W$ be the matrix defined in
  equation~\eqref{eq:W_def}.

  \begin{enumerate}
  \item \label{item:Wpositive}
    If $W \geq 0$,  then $\lambda(\alpha)$ achieves
    its global minimal value at  $\alpha=\alpha^\circ$.
  \item If $W \leq 0$, then $\lambda(\alpha)$ achieves
    its global maximal value at
    $\alpha=\alpha^\circ$.
  \end{enumerate}
\end{theorem}

We conjecture that $W \geq 0$ is also a \emph{necessary} condition for
the global minimum, and analogously for the global maximum.  In
Section~\ref{example:Kuc} we present an example that has a local
minimum that is not a global minimum; in this case $\Re W > 0$ while
$W$ is sign-indefinite.

If we additionally assume that the periodic operator $H$ is real
symmetric (has ``time-reversal symmetry'' in physics terminology),
there are certain points in the Brillouin zone that are critical for
every $\lambda$.  These are the points $\alpha^*\in \TT^d$ such that
$\cc{T(\alpha)} = T(\alpha^*-\alpha)$ for all $\alpha \in \TT^d$.  We
denote the set of these points by $\CornerP$ and refer to them
informally as ``corner points''; for the square parameterization
$(-\pi,\pi]^d$ of the Brillouin zone used throughout the paper, we
have $\CornerP = \{0, \pi\}^d$.

\begin{theorem}
  \label{thm:mainII}
Suppose, in addition to the hypotheses of Theorem \ref{thm:mainI}, that $H$ is real, and $\alpha^\circ\in \TT^d$ is a local extremum of $\lambda(\alpha)$. Then, in each of the following circumstances,
  $\lambda(\alpha^\circ)$ is the {global} extremal value:
  \begin{enumerate}
  \item If $\alpha^\circ \in \CornerP$.
  \item If $d \leq 2$.
  \item If $d = 3$ and the extremum is non-degenerate.
  \end{enumerate}
\end{theorem}

We therefore envision the following application of
Theorems~\ref{thm:mainI} and \ref{thm:mainII}.  In the setting of
Theorem~\ref{thm:mainI}, a gradient descent search for a local minimum
of $\lambda(\alpha)$ is to be followed by a computation of $W$,
using equation~\eqref{eq:W_def}.  If $W$ is non-negative,
Theorem~\ref{thm:mainI} guarantees that the global minimum has been found.
If $W$ is sign-indefinite, our conjecture requires the search to
continue.  In the setting of Theorem \ref{thm:mainII} one should first
check if any of the corner points $\CornerP$ is a local minimum,
possibly followed by the general gradient descent search.  But in any
of the cases specified in the theorem, \emph{the search can stop at
  the first local minimum found}, without having to compute the matrix $W$.

We now comment on the assumptions of our theorems.  One crossing edge
per generator is a substantial but common assumption: even for
$\mathbb{Z}^1$-periodic graphs with real symmetric $H$, the well-known
Hill Theorem fails in the presence of multiple crossing
edges, see \cite{ExnKucWin_jpa10}.  The restriction on the dimension
in Theorem~\ref{thm:mainII} is also essential: in dimension $d=4$ and
higher an internal point may be a local but not a global
extremum.  In Section~\ref{sec:examples} we provide such an example.  (Since
Theorem~\ref{thm:mainI} is valid for any $d$, it follows that the
corresponding $W$ is sign indefinite.)

\subsection*{Ideas of the proof and outline of the paper}
The assumptions of Theorems~\ref{thm:mainI} and \ref{thm:mainII} allow eigenvectors to
vanish on one side of a crossing edge. This situation is frequently
encountered in examples, as will be seen in Section~\ref{sec:eig}, but the proofs are significantly more 
complicated, since the matrix $\Omega$ in \eqref{eqn:BandOmega_def} is degenerate in that case. Here we 
give an overview of the paper and illustrate the proof of Theorem~\ref{thm:mainI}(1) when $\Omega$ is invertible. This greatly simplifies the statements and proofs of many of our results; see Remark~\ref{rem:invertible} for further discussion.


In Section~\ref{sec:basic defs} we introduce notation and clarify our 
assumptions on the structure of $\Gamma$. Next, in Section~\ref{sec:var
  flas}, we derive the first and second variation formulas \eqref{eq:derivatives_lambda}. A crucial observation is that the operator $W$ in \eqref{eq:W_def}, whose real part is the Hessian of $\lambda$, has the structure of a generalized Schur complement\,---\,``generalized'' because of the need to use the pseudoinverse in \eqref{eq:W_def}.

In Sections~\ref{sec:decomp} and \ref{sec:weyl}, we decompose the
operator $T(\alpha)$ as $T(\alpha) = S + R(\alpha) + \lambda(\alpha^\circ)$, where $S$ has a zero eigenvalue 
and does not depend on $\alpha$, and $R(\alpha)$ is a rank-$d$
perturbation\footnote{$R(\alpha)$ corresponds to $\sum R_j(\alpha_j)$
  in equation \eqref{eq:TS_decomposition}.} with the same
signature as $\Omega$. The rank is a consequence
of the ``one crossing edge per generator'' assumption. 
This decomposition allows us to establish a global Weyl-type bound
for the eigenvalues of $T(\alpha)$ in terms of eigenvalues of $S$, 
Lemma~\ref{lem:combinedWeyl}. If we further assume that $\Omega$ 
is positive, this simplifies to
\begin{equation}
\label{simpleWeyl}
	\lambda_n\big(T(\alpha)\big) \geq \lambda(\alpha^\circ) + \lambda_n(S)
\end{equation}
for all $\alpha \in \TT^d$.

Next, in Section~\ref{sec:index thy} we use a generalized
Haynsworth formula (see Appendix \ref{sec:gen_Haynsworth}) to relate 
the indices of $S$, $T(\alpha)$, $\Omega$ and the generalized Schur
complement $W$. Again assuming $\Omega$ is positive, 
the relationship simplifies to
\begin{equation*}
  i_-(W) = i_-(S) - i_-\big(T(\alpha^\circ) - \lambda(\alpha^\circ) \big),
\end{equation*}
where $i_-$ denotes the number of negative eigenvalues, i.e. the Morse
index. This can be expressed in words as ``the Morse index of $W$
equals the spectral shift between $S$ and the positive perturbation
$S+ R(\alpha^\circ) = T(\alpha^\circ) - \lambda(\alpha^\circ)$.''  This idea is further
developed for general self-adjoint operators in \cite{BerKuc_prep20},
where it is called the ``Lateral Variation Principle.''

To complete the proof of Theorem~\ref{thm:mainI}, in Section~\ref{sec: proof of main thm} we observe that
$W \geq 0$ implies
$i_-(W)=0$, and hence $\lambda(\alpha^\circ)$ saturates the lower
global Weyl bound in \eqref{simpleWeyl}. More precisely, we have
\[
	i_-(S) = i_-\big(T(\alpha^\circ) - \lambda(\alpha^\circ) \big) = n-1,
\]
where the second equality holds because $\lambda(\alpha^\circ)$ is the $n$-th eigenvalue of $T(\alpha^\circ)$. Since it was already observed that $0$ is an eigenvalue of $S$, this means $\lambda_{n}(S) = 0$. Substituting this into \eqref{simpleWeyl} gives $\lambda_n\big(T(\alpha)\big) \geq \lambda(\alpha^\circ)$ for all $\alpha$, as was to be shown. For a general non-degenerate (not necessarily positive) $\Omega$ the formulas are more complicated due to the presence of $i_-(\Omega)$, but the idea of the proof is identical. On the other hand, when $\Omega$ is degenerate we need to project away from its null space, and the proof is more involved.

In Section~\ref{sec:proof_mainII} we give the proof of Theorem~\ref{thm:mainII}.
The additional assumption of real symmetric $H$ implies $\Re W = W$ if
$\alpha^\circ \in \CornerP$, and so $W$ is completely determined by $\Hess \lambda(\alpha^\circ) = 2 \Re W$.  On the other hand, if $\alpha^\circ \not\in \CornerP$, then
$W$ may be complex. In this case we show that $\det W = 0$; this
allows us to estimate the spectrum of $W$ from the spectrum of
$\Re W$, but only in low dimensions.

Finally, in Section \ref{sec:examples} the main results are illustrated
with examples such as the honeycomb and Lieb lattices. We give examples where some components of the eigenvectors vanish, and conjecture that, under the hypotheses of Theorem~\ref{thm:mainI}, $W \geq 0$ is also a necessary condition for $\alpha^\circ$ to be a global minimum, and similarly for a maximum. We also provide (counter)examples showing that
when our assumptions are violated the theorems no longer hold. Specifically, we show that both Theorems~\ref{thm:mainI} and~\ref{thm:mainII} can fail if there are multiple crossing edges per generator, and Theorem~\ref{thm:mainII} no longer holds when $d>3$.

\section*{Acknowledgements} 
The authors gratefully acknowledge the American Institute of
Mathematics SQuaRE program from which this work developed.
G.B. acknowledges the support of NSF Grant DMS-1815075.
Y.C. acknowledges the support of NSF Grant DMS-1900519 and the Alfred
P. Sloan Foundation. G.C. acknowledges the support of NSERC grant
RGPIN-2017-04259. JLM acknowledges the generous support of NSF Grants
DMS-1352353 and DMS-1909035 as well as MSRI for hosting him as part of
this work was completed.  We are grateful to Peter Kuchment and Mikael Rechtsman
for many useful conversations about Bloch band structure, and John
Maddocks for kindly recalling to us the history of his generalized
Haynsworth formula.  We thank Lior Alon, Ram Band, Ilya Kachkovsky,
Stephen Shipman and Frank Sottile for interesting discussions and
insightful suggestions, and also thank the referees for their careful reading 
and helpful comments on the manuscript.


\section{Basic definitions and local behavior of $\lambda(\alpha)$}
\label{sec:prelim}

In this section we introduce a  matrix representation for the Floquet--Bloch
transformed operator $T(\alpha)$ (Section \ref{sec:basic defs}),
present a version of the Hellmann--Feynman variational formulas for
the $n$-th eigenvalue branch  $\lambda_n\big(T(\alpha)\big)$ (Section \ref{sec:var flas}), and give a decomposition formula for $T(\alpha)$ that works under the ``one crossing edge per generator'' assumption (Section \ref{sec:decomp}).

\subsection{Basic definitions} \label{sec:basic defs}

In this section we introduce a matrix representation for the Floquet--Bloch transformed operator. 
To do this we first present the notation we shall use for the vertices 
of the graph and the generators of the group action.

\begin{figure}
  \centering
  \includegraphics[scale=0.8]{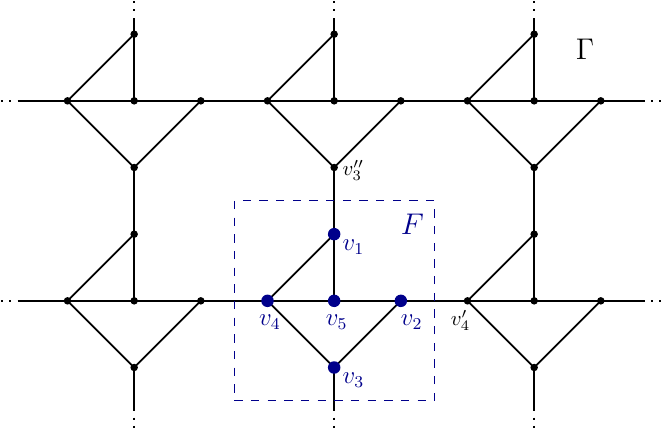}
  \caption{An example of a $\mathbb{Z}^2$-periodic graph $\Gamma$ and
    its fundamental domain $F$.  If $g_1$ and $g_2$ are the horizontal
    and vertical shifts generating the $\mathbb{Z}^2$ symmetry, then
    $v_4' =g_1v_4$ and $v_3'' = g_2v_3$.  The edges with
    end-vertices $(v_2, v'_4)$ and $(v_1, v_3'')$ give rise to the
    crossing edges, which are $(v_2, v_4)$ and $(v_1,v_3)$.
   }
  \label{fig:periodic_example}
\end{figure}

\begin{definition}
  \label{def:periodic_graph}
  A \emph{$\mathbb{Z}^d$-periodic graph} $\Gamma = (V,\sim)$ is a locally
  finite graph with a faithful cofinite group action by the free
  abelian group $G=\mathbb{Z}^d$.
\end{definition}

\noindent In this definition, $V$ is the set of vertices of the graph, and $\sim$ denotes the
adjacency relation between vertices. It will be notationally convenient
to postulate that $v\sim v$ for any $v\in V$.   Each
vertex is adjacent to finitely many other vertices (``locally
finite'').  Any $g\in G$ defines a bijection $v\mapsto gv$ on $V$
which preserves adjacency: $gu \sim gv$ if and only if $u\sim v$
(``action on the graph'').  For any $g_1,g_2 \in G$ we have
$g_1(g_2v) = (g_1g_2)v$ (``group action''). Also, $0\in G$ is the only
element that acts on $V$ as the identity (``faithful'').  The orbit
of $v$ is the subset $\{gv\colon g\in G\} \subset V$ and we assume
that there are only finitely many distinct orbits in $V$
(``cofinite'').

The ``one crossing edge per generator" assumption, introduced in
Definition \ref{def:crossing_edge}, is our central assumption on the
graph $\Gamma$.  In addition to the examples of
Figure~\ref{fig:more_lattices}, the graph of \cite{HarKucSob_jpa07} in
Figure~\ref{fig:periodic_example} also satisfies the assumption.  One
can think of such graphs as having been obtained by decorating $\ZZ^d$
by ``pendant'' or ``spider'' decorations
\cite{SchAiz_lmp00,DoKucOng_exner17}.  The terminology ``one crossing
edge per generator'' comes from the following consideration.
Definition \ref{def:crossing_edge} implies the existence of a choice of
$d$ generators $\{g_j\}_{j=1}^d$ of $G$ such that the fundamental
domain is connected only to its nearest neighbors with respect to the
generator set.  Namely,
\begin{equation} \label{eq:nn_fund}
  u\sim gv,\ \;\;u,v \in F
  \quad \Longrightarrow \quad
  g\in \{\mathrm{id}\}\cup\{g_j\}\cup\{g_j^{-1}\}.
\end{equation}
Conversely (because the graph is connected), for any generator $g_j$
in $\{g_j\}_{j=1}^d$, there is a unique pair of vertices
$u_j,v_j \in F$ such that $u_j \sim g_j v_j$.  The pair $(u_j,v_j)$
will be referred to as the \emph{$j$-th crossing edge. }
We note that while the
vertices $u_j$ and $v_j$ may not be adjacent in $\Gamma$, they will
become adjacent after the Floquet--Bloch transform, which we describe
next.  We also note that $u_j$ and $v_j$ may not be distinct.

Let $H$ be a periodic self-adjoint operator on $\ell^2(\Gamma)$. In
the present setting\footnote{Self-adjointness of more general graphs
  with Hermitian $H$ was studied in
  \cite{CdVTorTru_afstm11,Mil_ieot11}.} 
\begin{equation}
  \label{eq:Hdef}
  (H f)_u = \sum_{v\sim u} H_{u,v} f_v,
  \qquad H_{u,v} \in \mathbb{C}, \qquad H_{v,u} = \overline{H_{u,v}},
\end{equation}
and
\begin{equation}
  \label{eq:Hperiodic}
  H_{gu,gv} = H_{u,v}
  \quad\mbox{for any }u,v \in V, \ g\in G.
\end{equation}
We also assume that if $u,v$ are adjacent distinct vertices, then
$H_{u,v} \neq 0$.  Together with \eqref{eq:Hdef}, this means that there
is a non-zero interaction between vertices if and only if there is an
edge between them.
%

For a graph with one crossing edge per generator, the transformed
operator $T$ is a parameter dependent self-adjoint operator $T(\alpha)
\colon \ell^2(F) \to \ell^2(F)$, $\alpha \in \mathbb{T}^d$, acting as
\begin{equation}
  \label{eq:FB_def}
  (T(\alpha) f)_u = \sum_{\substack{g\in G,\,v\in F\\ gv \sim u}}
  H_{u,gv} \chi_\alpha(g) f_{v},
\end{equation}
where $F$ is a fundamental domain and
\begin{equation}
  \label{eq:character_def}
  \chi_\alpha(g) =
  \begin{cases}
    1 & \mbox{if } g = \mathrm{id}, \\
    e^{\pm i\alpha_j} & \mbox{if } g=g_j^{\pm1}.
  \end{cases}
\end{equation}
The function $\chi_\alpha$ is the character of a
representation of $G$; we do not need to list its values on the rest
of $G$ because of condition~\eqref{eq:nn_fund}. Continuing to denote by $N$ the number of vertices in a fundamental domain, this means that $T(\alpha)$ may be thought of as an $N\times N$ matrix. For a more general
definition of the Floquet--Bloch transform on graphs we refer the reader
to \cite[Chap.~4]{BerKuc_graphs}.

\begin{figure}
  \centering
  \includegraphics{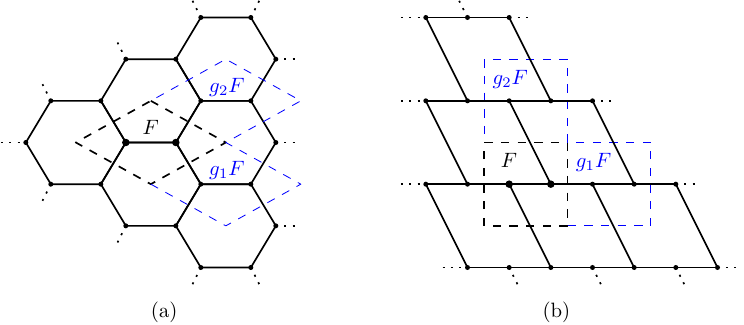}
  \caption{(a) ``Honeycomb'' embedding and (b) ``square'' embedding of
    the same graph into $\RR^2$.  The definition of the Floquet--Bloch
    transform in the physics literature usually takes the geometry of the
    embedding into account, but the resulting $T(\alpha)$ only differs
    by applying a linear transformation to the variables $\alpha$.}
  \label{fig:more_graphene}
\end{figure}

\begin{remark}
  \label{rem:brillouin_basis}
  It is important to note that we view a periodic graph as a topological
  object, with an abstract action by an abelian group.  In physical
  applications there is usually a natural geometric embedding of the
  graph into $\mathbb{R}^d$ and a geometric representation of the
  periodicity group (``lattice'').  The lattice, in turn, determines a
  particular parameterization of the Brillouin zone $\TT^d$ via the
  ``dual lattice.''  This physical parameterization may differ from
  the ``square lattice''
  parameterization~\eqref{eq:FB_def}--\eqref{eq:character_def} by a
  linear change in variables $\alpha$, as illustrated in
  Figure~\ref{fig:more_graphene}.  Our results do
  not depend on the choice of variables\,---\,in particular, the test matrix
  $W$ can be computed using any parameterization, see Lemma \ref{lem:invariance} below.
\end{remark}

\subsection{Variational formulas for $\lambda(\alpha)$}\label{sec:var flas}


Let $T(\alpha)$ be a real analytic family of $N \times N$ Hermitian
matrices, parametrized by $\alpha \in \TT^d$.  Fix a point
$\alpha^\circ \in \TT^d$ and suppose the $n$-th eigenvalue
$ \lambda_n\big(T(\alpha^\circ)\big)$ is simple, with
eigenvector $f^\circ$. For $\alpha$ in a neighborhood of $\alpha^\circ$, $ \lambda_n\big(T(\alpha)\big)$ is simple, and the function $\alpha \mapsto \lambda_n\big(T(\alpha)\big)$ is real analytic; see
\cite[Section II.6.4]{K76}. To streamline notation, we will denote this
function by $\lambda(\alpha)$.  We are interested in computing the
gradient and Hessian of $\lambda(\alpha)$ at $\alpha = \alpha^\circ$.

Let us introduce some notation and conventions. For a smooth enough
scalar function $u(\alpha)$ on $\TT^d$, its gradient, $\nabla u$, is a
column vector of length $d$, its differential, $D u$, is a row vector
of length $d$, and its Hessian, $\Hess u$, is a $d \times d$ symmetric
matrix.  For vector-valued functions we define $D$ componentwise: if
$f \colon \TT^d \to \mathbb R^N$, then $Df$ is an $N \times d$
matrix-valued function.  According to this convention, the matrix $B$
introduced in \eqref{eqn:BandOmega_def} is the $N \times d$ matrix
\begin{equation}
  \label{cdef} 
  B  := D \big(T(\alpha) f^\circ \big)\Big|_{\alpha^\circ}  = \left[
    \begin{array}{ccc}
      \vdots
      & \vdots
      & \vdots \\
      \frac{\p}{\p \alpha_1} (T(\alpha) f^\circ )\Big|_{\alpha^\circ}
      & \hdots
      & \frac{\p}{\p\alpha_d} (T(\alpha) f^\circ )\Big|_{\alpha^\circ} \\ 
      \vdots & \vdots & \vdots 
    \end{array}
  \right],
\end{equation}
where each $\frac{\p}{\p \alpha_j} (T(\alpha) f^\circ
)|_{\alpha^\circ}$ is a column vector of size $N$.  We stress that
$f^\circ$ remains fixed when the derivatives are taken with respect to $\alpha$.
We denote by $B^*$ the adjoint of $B$.  

We will regularly use the {\it Moore--Penrose pseudo-inverse} of a
matrix $A$, denoted $A^+$.  If $A$ is Hermitian, it can be computed as
\begin{equation}
  \label{eq:pinv_def}
  A^+  h = \sum_{\lambda_k(A) \neq 0} \frac{1}{\lambda_k(A)}
  \langle h, f_k \rangle f_k,
\end{equation}
where $\{ f_k \}$ is an orthonormal eigenbasis of $A$, with
corresponding eigenvalues $\{\lambda_k\}$.  With these terms defined,
we now state a multi-parameter version
  of the well
known Hellmann--Feynman eigenvalue variation formulas.

\begin{lemma}
  \label{lem:HF}
Let $T(\alpha)$ be an analytic family of $N \times N$ Hermitian
matrices, parametrized over $\alpha \in \TT^d$.
  Let $\lambda(\alpha^\circ)$ be a simple eigenvalue of $T(\alpha^\circ)$, and let
  $f^\circ$ be the corresponding normalized eigenvector.
For $B$ and $W$ defined in~\eqref{eq:W_def} and~\eqref{eqn:BandOmega_def}, respectively, we have
  \begin{equation}
    \label{Dlambda}
    \nabla \lambda(\alpha^\circ)
    = D\left\langle f^\circ, T(\alpha) f^\circ \right\rangle
    \Big|_{\alpha=\alpha^\circ} = B^* f^\circ,
  \end{equation}
  and
  \begin{equation}\label{Hesslambda}
    \Hess \lambda(\alpha^\circ) = 2 \Re W.
  \end{equation}
\end{lemma}

Since it is already known that $\lambda(\alpha)$ is analytic, the proof simply consists of using the well-known one-parameter version of the Hellmann--Feynman formula to compute directional derivatives. We include the details here for completeness.

\begin{proof}
For fixed $\eta \in \RR^d$ define $\hat\lambda(s) = \lambda(\alpha^\circ + s\eta)$, so that
\[
	\frac{d \hat\lambda}{ds}(0) = \left< \nabla \lambda(\alpha^\circ), \eta \right>.
\]
On the other hand, the one-dimensional Hellmann--Feynman formula (see
\cite[Remark II.2.2 (p. 81)]{K76}) says
\[
  \frac{d \hat\lambda}{ds}(0) = \big\langle f^\circ, T^{(1)} f^\circ \big\rangle,
\]
where
\[
  T^{(1)} f^\circ = \frac{d}{ds} T(\alpha^\circ + s\eta) f^\circ \Big|_{s=0} = B \eta.
\]
It follows that $\left< \nabla \lambda(\alpha^\circ), \eta \right> = \left< B^* f^\circ, \eta\right>$ for all $\eta$, which proves \eqref{Dlambda}.

Computing similarly for the second derivative, again using \cite[Remark II.2.2]{K76}, we find that
\[
	\left<\eta, [\Hess \lambda(\alpha^\circ)] \eta \right> = 2 \left[\big\langle f^\circ, T^{(2)} f^\circ \big\rangle - \big\langle T^{(1)} f^\circ, \big(T(\alpha^\circ)-\lambda(\alpha^\circ)\big)^+ T^{(1)} f^\circ \big\rangle \right],
\]
where
\[
	 \big\langle f^\circ, T^{(2)} f^\circ  \big\rangle = \frac12 \frac{d^2}{ds^2}  \big\langle f^\circ, T(\alpha^\circ + s\eta) f^\circ \big\rangle \Big|_{s=0} = \left<\eta, \Omega \eta \right>.
\]
Substituting $T^{(1)} f^\circ = B \eta$, it follows that
\[
	\left<\eta, [\Hess \lambda(\alpha^\circ)] \eta \right> = 2 \left<\eta, (\Omega - B^*\big(T(\alpha^\circ)-\lambda(\alpha^\circ)\big)^+ B )\eta \right> = 2 \left< \eta, W \eta \right>
\]
for all $\eta \in \RR^d$, and hence the symmetric parts of the matrices $\Hess \lambda(\alpha^\circ)$ and $2W$ coincide:
\[
	\Hess \lambda(\alpha^\circ) + \Hess \lambda(\alpha^\circ)^T = 2 (W + W^T).
\]
Since the Hessian is real and symmetric, and $W$ is Hermitian, this simplifies to $\Hess \lambda(\alpha^\circ) = W + \cc{W} = 2 \Re W$, as claimed.
\end{proof}

We conclude this section by verifying the claim made in Remark
\ref{rem:brillouin_basis}, that the sign of $W$ used in
Theorem~\ref{thm:mainI} can be computed using any parameterization of
the torus.

\begin{lemma}\label{lem:invariance}
  Let $\phi \colon \TT^d \to \TT^d$ be a diffeomorphism, and define
  $\widetilde T(k) = T \big(\phi(k)\big)$.  Let $\alpha=\alpha^\circ$
  be a critical point of a simple eigenvalue
  $\lambda_n\big(T(\alpha)\big)$.  For the matrix $W$ computed from
  $T(\alpha)$ at $\alpha^\circ$ according to \eqref{eq:W_def}, and $\widetilde W$
  similarly computed from $\widetilde T(k)$ at
  $k^\circ := \phi^{-1}(\alpha^\circ)$, we have
  \begin{equation}
    \label{eq:Winvariance}
    \widetilde W = J^T W J,
  \end{equation}
  where $J$ is the real invertible Jacobian matrix $J =
  D\phi(k)\big|_{k=k^\circ}$.
\end{lemma}

\begin{proof}
  Applying the chain rule to the definition of $\widetilde B$, we get
  \begin{equation*}
    \widetilde B 
    := D\left(\widetilde T(k) f^\circ\right)\Big|_{k = k^\circ}
    = D\big(T(\alpha) f^\circ \big) 
    \Big|_{\alpha=\alpha^\circ} D\phi(k)\big|_{k=k^\circ} 
    = B J.
  \end{equation*}
  In particular, since $\alpha^\circ$ is a critical point,
  $\widetilde B^* f^\circ = J^T B^* f^\circ = 0$, cf.\
  equation~\eqref{Dlambda}.  Therefore $k^\circ$ is a
  critical point of the simple eigenvalue
  $\lambda_n\big(\widetilde T(k)\big)$.  By a similar calculation,
  $\alpha^\circ$ is a critical point of the scalar function
  $\Phi(\alpha) := \left\langle  f^\circ, T(\alpha) f^\circ
  \right\rangle$. The Hessian at a critical point transforms under a
  diffeomorphism as
  \begin{align}\label{eq:diffHess}
    \Hess \Phi(\alpha(k)) \big|_{k = k^\circ}
    = J^T \left( \Hess
    \Phi(\alpha)
    \big|_{\alpha = \alpha^\circ} \right) J,
  \end{align}
  implying $\widetilde \Omega = J^T \Omega J$.
  Putting it all together gives \eqref{eq:Winvariance}.
\end{proof}

  We remark that since $J$ is real, \eqref{eq:Winvariance} implies
  $\Re \widetilde W = J^T (\Re W)J$. This could also have been obtained by applying the transformation rule \eqref{eq:diffHess} to the function $\lambda(\alpha)$, which has Hessian proportional to $\Re W$, according to Lemma~\ref{lem:HF}.
  

\subsection{The decomposition of $T(\alpha)$}
\label{sec:decomp}

Lemma~\ref{lem:HF} is valid for any family $T(\alpha)$ of Hermitian
matrices. We now consider the specialized form of the $T(\alpha)$
appearing as the Floquet--Bloch transform of a graph \emph{with one
  crossing edge per generator}.  For a graph satisfying Definition
\ref{def:crossing_edge}, there exists a choice of fundamental domain
and periodicity generators such that the Floquet--Bloch transformed
operator $T(\alpha)$ is given by equation~\eqref{eq:FB_def} and the
Brillouin zone $\TT^d$ is parameterized by $\alpha \in (-\pi,\pi]^d$.
Other physically relevant parameterizations of $T(\alpha)$ may be
obtained by a change of variables $\alpha$; by Lemma \ref{lem:invariance}, it is enough
to establish our theorems for a single parameterization.

The operator $T(\alpha)$ defined by \eqref{eq:FB_def} can be decomposed as
\begin{equation}
  \label{eq:Talphadecomp}
  T(\alpha) = T_0 + \sum_{j=1}^d T_j(\alpha_j),
\end{equation}
where $T_0$ is a constant Hermitian matrix, and each $T_j$ has at
most two nonzero entries.  More precisely, if $\{g_j\}_{j=1}^d$ are the generators for $G$, $(u_j,v_j)$ is the
$j$-th crossing edge (see Section~\ref{sec:basic defs}) and 
$$h_j := H_{u_j,g_j v_j},$$ then
\begin{equation}
  \label{eq:Tj_def}
  T_j(\alpha_j)
  = h_j e^{i \alpha_j} \Eb_{u_j, v_j} + \cc{h_j} e^{-i \alpha_j} \Eb_{v_j, u_j},
\end{equation}
where $\Eb_{u,v}$ denotes the $N \times N$ matrix with $1$ in the $u$-$v$ entry and all other entries equal to $0$. If $u_j \neq v_j$, then $T_j(\alpha_j)$ will have two nonzero entries, appearing in a $2 \times 2$ submatrix of the form
\[
  \begin{bmatrix}
    0 & h_j e^{i \alpha_j} \\
    \cc{h_j} e^{-i \alpha_j} & 0
  \end{bmatrix} .
\]
If $u_j = v_j$, then
$T_j(\alpha_j)$ has a single nonzero entry,
namely $2 \Re \left(h_j e^{i \alpha_j}\right)$, on the diagonal.

We now give explicit formulas for $B$, $\Omega$, and their combinations
that will be useful later.

\begin{lemma}
  \label{lem:Bexplicit}
  Let $T(\alpha)$ be as in \eqref{eq:Talphadecomp}. Then for $j=1, \dots, d$, the matrix $B$
  defined in \eqref{cdef} has $j$-th column
  \begin{equation}
    \label{eq:Bexplicit}
    \col_j(B) = i \left( h_j e^{i \alpha^\circ_j} f^\circ_{v_j}
      \eb_{u_j} - \cc{h_j} e^{-i \alpha^\circ_j} f^\circ_{u_j} \eb_{v_j} \right),
  \end{equation}
  where $\{\eb_u\}_{u=1}^N$ denotes the standard basis  for $\bbC^N$.
  Consequently, by Lemma~\ref{lem:HF},
  \begin{equation*}
    \frac{\partial \lambda}{\partial \alpha_j} (\alpha^\circ)
    = -2 \Im(h_j e^{i \alpha^\circ_j}
    f^\circ_{v_j} \cc{f^\circ_{u_j}}),
  \end{equation*}
  and $\alpha^\circ$ is a critical point of $\lambda$ if and only if
  \begin{equation}
    \label{cp:real}
    h_j e^{i \alpha^\circ_j} f^\circ_{v_j} \overline{f^\circ_{u_j}} \in \RR
  \end{equation}
  for each $j=1,\ldots,d$.
\end{lemma}

It was already observed in \cite[Lemma A.2]{BanBerWey_jmp15} that
\eqref{cp:real} holds at a critical point; we include a proof here for convenience since it follows easily from
\eqref{eq:Bexplicit}.

\begin{proof}
Using \eqref{eq:Tj_def} we obtain 
 \[
 	T_j(\alpha_j) f^\circ = h_j e^{i \alpha_j} f^\circ_{v_j} \eb_{u_j} + \cc{h_j} e^{-i \alpha_j} f^\circ_{u_j} \eb_{v_j}
 \]
 for each $j$, and \eqref{eq:Bexplicit} follows.
Then, from \eqref{Dlambda} and \eqref{eq:Bexplicit}, we have
 \[
 	 \frac{\partial \lambda}{\partial \alpha_j} (\alpha^\circ) = \left< \col_j(B), f^\circ \right> =  i \left( h_j e^{i \alpha^\circ_j} f^\circ_{v_j} \overline{f^\circ_{u_j}} - \cc{h_j} e^{-i \alpha^\circ_j} f^\circ_{u_j} \overline{f^\circ_{v_j}} \right)
	 = -2 \Im(h_j e^{i \alpha^\circ_j}
    f^\circ_{v_j} \cc{f^\circ_{u_j}}),
\]
which completes the proof.
\end{proof}

\begin{lemma}
  \label{lem:Oexplicit}
  For $T(\alpha)$ as in \eqref{eq:Talphadecomp}, the matrix $\Omega$
  defined in \eqref{eqn:BandOmega_def} is diagonal, with
  \begin{equation}
    \label{Omegaexplicit}
    \Omega_{jj} = -\Re\left(h_j e^{i \alpha^\circ_j}
      f^\circ_{v_j} \overline{f^\circ_{u_j}}\right)
  \end{equation}
 for each $j=1,\ldots,d$.
\end{lemma}

\begin{proof}
As in the proof of Lemma \ref{lem:Bexplicit}, we compute
\[
	\left< T_j(\alpha_j) f^\circ, f^\circ \right> = h_j e^{i \alpha_j} f^\circ_{v_j} \overline{f^\circ_{u_j}} + \cc{h_j} e^{-i \alpha_j} \overline{f^\circ_{v_j}} f^\circ_{u_j} 
	= 2 \Re \left(h_j e^{i \alpha_j} f^\circ_{v_j} \overline{f^\circ_{u_j}}\right),
\]
and the result follows.
 \end{proof}

If $\alpha^\circ$ is a critical point, \eqref{cp:real} and \eqref{Omegaexplicit} together imply that for each $j=1,\ldots,d$,
\begin{equation}
\label{Omegaexplicit2}
	\Omega_{jj}
	= - h_j e^{i \alpha^\circ_j} f^\circ_{v_j} \overline{f^\circ_{u_j}}
	= - \cc{h_j} e^{- i \alpha^\circ_j} \overline{f^\circ_{v_j}} f^\circ_{u_j}.
\end{equation}

In what follows we let $J'$ denote the indices of non-zero diagonal entries of $\Omega$, and
let $J''$ be its complement, namely
\begin{equation}
  \label{eq:Jprime_def}
  J' := \{j : f^\circ_{u_j} f^\circ_{v_j} \neq 0\},
  \qquad
  J'' := \{j : f^\circ_{u_j} f^\circ_{v_j} = 0\}.
\end{equation}

\begin{lemma}
  \label{lem:BOBexplicit}
  Let $P=P_{\Null(\Omega)}$ be the orthogonal projection onto
  $\Null(\Omega)$.  If $\alpha^\circ$ is a critical point of
  $\lambda(\alpha)$, then
  \begin{align}\label{eq:BOB}
    B \Omega^+ B^*
    &= \sum_{j \in J'} \left(
      \frac{\Omega_{jj}}{|f^\circ_{u_j}|^2} \Eb_{u_j,u_j}
      + h_j e^{i \alpha^\circ_j} \Eb_{u_j,v_j}
      + \cc{h_j} e^{-i \alpha^\circ_j} \Eb_{v_j,u_j}
      +\frac{\Omega_{jj}}{|f^\circ_{v_j}|^2} \Eb_{v_j,v_j}\right),\\
    \label{eq:BPB}
    BPB^* 
    &= \sum_{j \in J''} |h_j|^2 \left( | f^\circ_{v_j}|^2 \Eb_{u_j,u_j}
      + | f^\circ_{u_j}|^2 \Eb_{v_j,v_j} \right).
  \end{align}
  Therefore, $\Ran(BPB^*)$ is spanned by the vectors
  \begin{equation}\label{RanBPB}
    \big\{ \eb_{u_j} : f^\circ_{u_j} = 0, f^\circ_{v_j} \neq 0 \big\}
    \cup \big\{ \eb_{v_j} : f^\circ_{v_j} = 0, f^\circ_{u_j} \neq 0 \big\}.
  \end{equation}
\end{lemma}

\begin{remark}
If $u_j = v_j$, the $j$th summand in \eqref{eq:BOB} is identically zero; otherwise it contains a nonzero $2\times2$ submatrix of the form
\[
  \begin{bmatrix}
    \Omega_{jj} |f^\circ_{u_j}|^{-2} & 
    h_j e^{i \alpha^\circ_j} \\[5pt]
    \cc{h_j} e^{-i \alpha^\circ_j} &
    \Omega_{jj} |f^\circ_{v_j}|^{-2}
  \end{bmatrix}.
\]
The off-diagonal part is precisely the matrix $T_j(\alpha^\circ_j)$ appearing in \eqref{eq:Talphadecomp}; this fact is essential to the proof of Lemma \ref{lem:combinedWeyl} below.
\end{remark}

\begin{proof}
The pseudoinverse $\Omega^+$ is diagonal, with
\[
	(\Omega^+)_{jj} = \begin{cases} \Omega_{jj}^{-1} , & j \in J', \\
	0, & j \in J''.
	\end{cases}
\]
It follows that
\[
	B \Omega^+ B^* = \sum_{j \in J'} \Omega_{jj}^{-1} \col_j(B) \col_j(B)^*.
\]
Using \eqref{eq:Bexplicit} for $\col_j(B)$ and~\eqref{Omegaexplicit2} for $\Omega_{jj}$, we obtain \eqref{eq:BOB}.

Similarly, the orthogonal projection $P$ onto $\Null(\Omega)$ is diagonal, with
\[
	P_{jj} = \begin{cases} 0, & j \in J', \\
	1, & j \in J'',
	\end{cases}
\]
and so
\[
	BPB^* = \sum_{j \in J''} \col_j(B) \col_j(B)^*.
\]
Again, using \eqref{eq:Bexplicit} for $\col_j(B)$, \eqref{eq:BPB} follows.

Finally, note that the $j$th summand in \eqref{eq:BPB} contains at most one nonzero term, since either $f^\circ_{u_j} = 0$ or $f^\circ_{v_j} = 0$ for each $j \in J''$. In particular, $BPB^*$ is diagonal, and the $u$th entry is nonzero if and only if either $u = u_j$ for some $j$ such that $f^\circ_{u_j} = 0$ and $f^\circ_{v_j} \neq 0$, or $u = v_j$ for some $j$ with $f^\circ_{v_j} = 0$ and $f^\circ_{u_j} \neq 0$. This establishes \eqref{RanBPB} and completes the proof.
\end{proof}

\section{Global properties of $\lambda(\alpha)$: Proof of Theorem~\ref{thm:mainI}}
\label{sec:shift}

According to Lemma~\ref{lem:HF}, the matrix $\Re W$ determines if $\lambda(\alpha)$ has a
\emph{local} extremum at a given critical point $\alpha^\circ$.  We now turn to
the proof of Theorem~\ref{thm:mainI}, which states that the
\emph{global} properties of $\lambda(\alpha^\circ)$ are determined by the
matrix $W$ itself\,---\,without taking its real part.  

The proof hinges on the fact that we can decompose\footnote{When
  $\Omega$ is invertible} $T(\alpha) = S + R(\alpha)$, where
$R(\alpha)$ is a rank-$d$ perturbation whose signature is determined
by $\Omega$.  This yields global bounds on the eigenvalues of
$T(\alpha)$, given in Lemma~\ref{lem:combinedWeyl}..  In subsequent sections we will show that if
$W$ is sign-definite at a critical point $\alpha^\circ$, then these global bounds become saturated and we thus have a global extremum,
proving Theorem~\ref{thm:mainI}.

\subsection{A Weyl bracketing for eigenvalues of $T(\alpha)$}
\label{sec:weyl}

Let us introduce some notation that will be of use.  The inertia of a
Hermitian matrix $M$ is defined to be the triple
\begin{equation}
  \label{eq:inertia1}
  \In(M) := (i_+(M), i_-(M), i_0(M))
  =: (i_+,i_-,i_0)_M
\end{equation}
of numbers of positive, negative, and zero eigenvalues of $M$
correspondingly.\footnote{This particular ordering appears to be
traditional in the literature.}  The second notation will be sometimes
used to avoid repetitive specification of the matrix $M$.

Define the subspace $\ssQ  \subseteq
\bbC^N$ by
\begin{equation}
  \label{eq:subspaceQ}
  \ssQ = \Null(BP_{\Null(\Omega)} B^*),
\end{equation}
and let $Q$ denote the orthogonal projection onto $\ssQ$.  For an
operator $A$, we denote by $\big(A\big)_\ssQ$ the operator $QAQ^*$
{\it considered as an operator on the vector space $\ssQ$}.  We
highlight that we consider this operator acting on $\ssQ$ in order to
make the dimensions arising in each of our statements below simple to
understand.  We now define
\begin{equation}
  \label{Sdef}
  S:= \left( T(\alpha^\circ)- \lambda(\alpha^\circ) - B \Omega^+ B^* \right)_\ssQ,
\end{equation}
and
\begin{equation}
  \label{eq:S-infinity-def}
  i_\infty(S) := N - \dim(\ssQ),
\end{equation}
where $B$ and $\Omega$ given by \eqref{eqn:BandOmega_def}.

\begin{remark}
  The subspace $\ssQ$ is defined in order to make $\Omega$
  invertible on $B^*(\ssQ)$.  If one considers
  $T(\alpha^\circ)- \lambda(\alpha^\circ) - B \Omega^{-1} B^*$ as a linear relation, then $\ssQ$ is
  its regular part and $i_\infty(S)$ is the dimension of its
 singular part.  Informally, $i_\infty(S)$ is the multiplicity of
  $\infty$ as an eigenvalue of $T(\alpha^\circ)- \lambda(\alpha^\circ) - B \Omega^{-1} B^*$.
\end{remark}

\begin{remark}
\label{rem:ifB}
It follows from the formula for $BPB^*$ given in \eqref{eq:BPB} that $i_\infty(S) = \rk(BPB^*)$ is the dimension of the vector space spanned by $\{ \col_j(B) : j \in J''\}$; see also \eqref{RanBPB}.
\end{remark}

\begin{remark}
\label{rem:invertible}
In the introduction we gave an outline of the paper assuming that the eigenvector $f^\circ$ 
is nowhere zero, and hence $\Omega$ is invertible. In that case the set $J''$ defined in \eqref{eq:Jprime_def} is empty and $P_{\Null(\Omega)}=0$.
As a result, the subspace $\ssQ$ is the entire space $\bbC^N$, and so $i_\infty(S)=0$.
We invite the reader to first read the following proofs with these stronger assumptions in place.
\end{remark}

We first observe that $S$ has a $0$ eigenvalue; this fact will be used in the proofs of 
Theorems~\ref{thm:mainI} and~\ref{thm:mainII}.

\begin{lemma}
\label{lem:CP_eig0}
If $\lambda(\alpha)= \lambda_n\big(T(\alpha)\big)$ has a critical point at $\alpha^\circ$ and  $\lambda(\alpha^\circ)$ is a simple eigenvalue,
then $0$ is an eigenvalue of $S$ as defined in \eqref{Sdef}.
\end{lemma}

\begin{proof}
  Lemma~\ref{lem:HF} implies $B^* f^\circ = 0$, so $f^\circ
  \in\ssQ$ and hence $$S f^\circ = \big(T(\alpha^\circ)- \lambda(\alpha^\circ)\big) f^\circ = 0.$$
\end{proof}

The main result of this subsection is the following Cauchy--Weyl
bracketing inequality between $S$ and $T(\alpha)$.

\begin{lemma}
  \label{lem:combinedWeyl}
	Suppose that $\lambda(\alpha)=\lambda_n\big(T(\alpha)\big)$ has a critical point at $\alpha^\circ$ and that $\lambda(\alpha^\circ)$ is a simple eigenvalue. Let $f^\circ$ be the corresponding eigenvector and assume that $f^\circ$ is non-zero on at least one end of any
  crossing edge (see Section~\ref{sec:basic defs}). Then, for any $\alpha\in\TT^d$ the eigenvalues of
  $T(\alpha)$ and $S$ are related by
  \begin{equation}
    \label{eq:Weyl}
    \lambda_{n - i_-(\Omega) - i_\infty(S)}(S)
    \leq \lambda_n\big(T(\alpha)\big) - \lambda(\alpha^\circ)
    \leq \lambda_{n + i_+(\Omega)}(S).
  \end{equation}
\end{lemma}

\begin{proof}
  We recall that the crossing edges for the graph are denoted by $(u_j,v_j)$ 
  with $j=1, \dots, d$ (see Section~\ref{sec:basic defs}).
  Let $J' = \{j : f^\circ_{u_j} f^\circ_{v_j} \neq 0\}$ and consider the matrix
  \begin{equation}
    \label{eq:TS_decomposition}
    S'(\alpha) := T(\alpha) - \lambda(\alpha^\circ)
    - \sum_{j\in J'} R_j(\alpha_j),
  \end{equation}
  with 
  \begin{equation}
    \label{Rjdef}
    R_j(\alpha_j) :=\frac{\Omega_{jj}}{|f^\circ_{u_j}|^2} \Eb_{u_j,u_j}
    + h_j e^{i \alpha_j} \Eb_{u_j,v_j}
    + \cc{h_j} e^{-i \alpha_j} \Eb_{v_j,u_j}
    + \frac{\Omega_{jj}}{|f^\circ_{v_j}|^2} \Eb_{v_j,v_j}.
  \end{equation}
  We note that at the point $\alpha=\alpha^\circ$ the sum of
  $R_j(\alpha_j^\circ)$ matches the expression for $B\Omega^+B^*$
  obtained in Lemma~\ref{lem:BOBexplicit}.
  If $u_j \neq v_j$, the matrix $R_j(\alpha_j)$ has four nonzero
  entries, appearing in a $2\times 2$ submatrix of the form
  \begin{equation}
    \label{eq:Rblock}
    \begin{bmatrix}
      \Omega_{jj} |f^\circ_{u_j}|^{-2}
      & h_j e^{i\alpha_j} \\
      \cc{h_j} e^{-i\alpha_j}
      & \Omega_{jj} |f^\circ_{v_j}|^{-2}
    \end{bmatrix}.
  \end{equation}
  If $u_j = v_j$, then $R_j(\alpha_j)$ has a single nonzero entry,
  \begin{equation}
    \label{eq:Rblock1}
    2 \Re \left( h_j e^{i \alpha_j} - h_j e^{i \alpha^\circ_j}
    \right),  
  \end{equation}
  appearing on the diagonal.
  
  The matrices $R_j(\alpha_j)$ have several crucial properties.
  First, they are the minimal rank perturbations that remove from
  $S'(\alpha)$ any dependence on the $\alpha_j$ with $j\in J'$.
  Second, once restricted to $\ssQ = \Null(BP_{\Null(\Omega)} B^*)$, the dependence on the remaining
  $\alpha_j$ is eliminated and $S'(\alpha)$ turns into $S$ defined in
  \eqref{Sdef}. More precisely, we will now show that
  \begin{equation}
    \label{eq:A0_from_Aprime}
    S = \big(S'(\alpha)\big)_\ssQ.
  \end{equation}

  From \eqref{eq:Tj_def}, \eqref{eq:BOB} and \eqref{Rjdef} we obtain
  \begin{align*}
    \sum_{j\in J'} R_j(\alpha_j)
    &= \sum_{j \in J'} \left[T_j(\alpha_j) - T_j(\alpha^\circ_j)\right]
      + B\Omega^+ B^* \\
    &= T(\alpha) - T(\alpha^\circ)
      - \sum_{j \notin J'} \left[ T_j(\alpha_j)
      - T_j(\alpha^\circ_j) \right]
      + B \Omega^+ B^*,
  \end{align*}
  and so
  \begin{equation}
    \label{Spdef1}
    S'(\alpha) = T(\alpha^\circ) - \lambda(\alpha^\circ) - B \Omega^+ B^* + \sum_{j \in J''} \left[ T_j(\alpha_j) - T_j(\alpha^\circ_j) \right],
  \end{equation}
  where $J'' = \{j : f^\circ_{u_j} f^\circ_{v_j} = 0\}$. Each of the
  summands $T_j(\alpha_j) - T_j(\alpha^\circ_j)$ is a linear
  combination of the basis matrices $\Eb_{u_j,v_j}$ and
  $\Eb_{v_j,u_j}$. Fix an arbitrary $j \in J''$. Since $f^\circ$ is non-zero on at least one 
  end of any crossing edge, we may assume without
  loss of generality that $f^\circ_{u_j} = 0$ and
  $f^\circ_{v_j} \neq 0$. Then from \eqref{RanBPB} we have
  $\eb_{u_j} \in \Ran(BPB^*) = \Null(BPB^*)^\perp = \ssQ^\perp$, so
  $Q \eb_{u_j} = 0$, where $Q$ is the  projection operator onto $\ssQ$. This implies $Q \Eb_{u_j, v_j} = 0$ and
  $\Eb_{v_j,u_j} Q^* = 0$ and therefore
  \[
    Q \Eb_{u_j, v_j} Q^* = Q \Eb_{v_j, u_j} Q^* = 0.
  \]
It follows that all the summands in \eqref{Spdef1} with $j \in J''$ vanish
  when conjugated by the projection matrix $Q$. This completes the
  proof of \eqref{eq:A0_from_Aprime}.
  
  We now relate the eigenvalues of $T(\alpha)$ and $S'(\alpha)$ by
  computing the signature of the $R_j(\alpha_j)$ perturbations.  If
  $u_j\neq v_j$, it follows from \eqref{Omegaexplicit2} that the
  determinant of the matrix \eqref{eq:Rblock} vanishes, and so it has
  rank one, with signature given by the sign of $\Omega_{jj}$.

  On the other hand, if $u_j=v_j$, the matrix has at most one non-zero
  entry.  From Lemma \ref{lem:Bexplicit} (equation~\eqref{cp:real}
  with $f^\circ_{v_j}=f^\circ_{u_j}$) we have
  $h_j e^{i \alpha^\circ_j} \in \RR$, and so
  \begin{align*}
    \Re \left(h_j e^{i \alpha_j} - h_j e^{i \alpha^\circ_j} \right) = h_j e^{i \alpha^\circ_j} \Re \left( e^{i (\alpha_j - \alpha^\circ_j)} - 1 \right) = h_j e^{i \alpha^\circ_j} [ \cos(\alpha_j - \alpha^\circ_j) - 1].
  \end{align*}
  Since $\cos(\alpha_j - \alpha^\circ_j) < 1$ for
  $\alpha_j \neq \alpha_j^\circ$ and
  $\Omega_{jj} = - h_j e^{i \alpha^\circ_j} |f^\circ_{u_j}|^2$, we
  conclude that $R_j(\alpha_j)$ has the same sign as $\Omega_{jj}$
  provided $\alpha_j \neq \alpha_j^\circ$.

  Summing over all $j \in J'$, we conclude that
  $T(\alpha)-\lambda(\alpha^\circ)-S'(\alpha)$ has at most $i_-(\Omega)$ negative
  and at most $i_+(\Omega)$ positive eigenvalues.  It follows from the
  classical Weyl interlacing inequality that
  \begin{equation}
    \label{eq:classicalWeyl}
    \lambda_{n - i_-(\Omega)}\big(S'(\alpha)\big)
    \leq \lambda_n\big(T(\alpha)\big)-\lambda(\alpha^\circ)
    \leq \lambda_{n + i_+(\Omega)}\big(S'(\alpha)\big)
  \end{equation}
  for 
  all $\alpha \in \TT^d$.

  Now, applying the Cauchy interlacing inequality (for submatrices or,
  equivalently, for restriction to a subspace) to $S'(\alpha)$ and
  $S = \big(S'(\alpha)\big)_\ssQ$, we get
  \begin{equation*}
    \lambda_{m - i_\infty(S)}(S)
    \leq \lambda_m\big(S'(\alpha)\big)
    \leq \lambda_{m}(S)
  \end{equation*}
  for all $\alpha \in \TT^d$. Combining this with
  \eqref{eq:classicalWeyl}, we obtain the result.
\end{proof}

\begin{remark}
  The hypothesis that $f^\circ$ does not vanish identically on any
  crossing edge, which was used in the proof of
  \eqref{eq:A0_from_Aprime}, can be weakened slightly. If
  $f^\circ_{u_j} = f^\circ_{v_j} = 0$ for some $j$, the proof would
  still hold if we can show that $\eb_{u_j}$ or $\eb_{v_j}$ belong to
  the range of $BPB^*$.  The latter would hold if there exists
  another index $k$ such that $u_k$ coincides with either $u_j$ or $v_j$,
  and $f^\circ_{v_k} \neq 0$.
\end{remark}

\subsection{Index formulas for $W$}\label{sec:index thy}
In this subsection we study the relationship between the index 
of $W$ and the indices we have already encountered, namely $i_-(\Omega)$,
$i_+(\Omega)$ and $i_\infty(S)$.  This is done by observing that
$W$ has the structure of a Schur complement and then using a suitably
generalized Haynsworth formula.

The following lemma applies to any matrices $A$, $B$ and $\Omega$ satisfying the given hypotheses. In Section~\ref{sec: proof of main thm} we will apply it specifically to $A = T(\alpha^\circ) - \lambda(\alpha^\circ)$, and $B$ and $\Omega$ from \eqref{eqn:BandOmega_def}.

\begin{lemma}
\label{lem:Windex}
Suppose $W = \Omega - B^* A^+ B$, where $\Omega$ and $A$ are Hermitian matrices of size $d\times d$ and $N\times N$, respectively, and $B$ is an $N \times d$ matrix satisfying
  \begin{equation}
    \label{eq:Wregular}
    \Null(A) \subset \Ran(B)^\perp = \Null(B^*).
  \end{equation}
  Let $P=P_{\Null(\Omega)}$ be the orthogonal projection onto
  $\Null(\Omega)$ and denote $\ssQ := \Null(BPB^*)$.  Define
  \begin{equation*}
	S := \left(A - B \Omega^+ B^*\right)_{\ssQ},
  \end{equation*}
  and
  \begin{equation}
    \label{eq:i_infty_def}
    i_\infty(S) := \rk (B P B^*) = N - \dim(\ssQ).
  \end{equation}
  Then,
  \begin{align}
    i_-(W) &= i_-(\Omega) +  i_-(S) + i_\infty(S) - i_-(A)
             \label{eq:i-W}, \\
    i_0\,(W) &= i_0\,(\Omega) + i_0 \,(S) - i_\infty (S) - i_0 \,(A),
             \label{eq:i0W} \\
    i_+(W) &= i_+(\Omega) + i_+(S) + i_\infty(S) - i_+(A)
             \label{eq:i+W} \\
           &= i_+(\Omega) - i_-(S) - i_0(S) + i_-(A) + i_0(A).
             \label{eq:i+Wbis}
  \end{align}
\end{lemma}

\begin{remark}
  \label{rem:Windex}
  If $\Omega$ is strictly positive, equation~\eqref{eq:i-W} simplifies
  to
  \begin{equation*}
    i_-(W) =  i_-(S) - i_-(A).
  \end{equation*}
 This can be expressed in words as ``the Morse index of $W$ is the spectral shift
  at $-0$ between $S$ and its positive perturbation
  $A = S+ B \Omega^{-1} B^*$.''  This idea is further developed in
  \cite{BerKuc_prep20}.
\end{remark}

\begin{remark}
  For $i_+(W)$ we have two forms: equation~\eqref{eq:i+W} is similar
  to the previous ones, but equation~\eqref{eq:i+Wbis} will be
  directly applicable in our proofs.  In addition,  the
  ``renormalized'' form \eqref{eq:i+Wbis}  (in the physics sense of cancelling
  infinities) is the one that retains its meaning if $S$ and $A$ are bounded
  below but unbounded above, as they would be in generalizing this
  result to elliptic operators on compact domains.
\end{remark}

\begin{proof}[Proof of Lemma~\ref{lem:Windex}]
  The definitions of the matrices $W$ and $S$ are
  reminiscent of the Schur complement, and so to investigate their
  indices, it is natural to use the Haynsworth formula~\cite{Hay_laa68}.  
  For a Hermitian matrix in block form, $M = \left[\begin{smallmatrix} A & B \\ B^* & C \end{smallmatrix} \right]$
  with $A$ \emph{invertible}, the Haynsworth formula  states
  that
  \begin{equation}
    \label{eq:Haynsworth}
    \In(M) = \In(A) + \In(C - B^* A^{-1} B),
  \end{equation}
  where the inertia triples add elementwise.
  Several versions of the formula are available for the case when $A$
  is no longer invertible (see \cite{Cot_laa74,Mad_laa88}), but the form most
  suitable for our purposes (equation~\eqref{eq:extendedHaynsworth} below) we could not find in
  the literature.  For completeness, we provide its proof in
  Appendix~\ref{sec:gen_Haynsworth}. Denote by $P_A$ the orthogonal projection onto
  the nullspace of $A$ and define
  \begin{equation}
    \label{eq:subspace_Q_iinf}
    \ssQ_A = \Null(B^*P_A B),
    \qquad
    i_\infty(M/A) = \rk(B^*P_A B) = \dim(C) - \dim(\ssQ_A),
  \end{equation}
  where $M/A$ is the generalized Schur complement of the block $A$,
\begin{equation}
  \label{eq:Schur_complement_alt}
  M/A := C - B^* A^{+} B.
\end{equation}
  Our generalized Haynsworth formula states that
  \begin{equation}
    \label{eq:extendedHaynsworth}
    \In(M) = \In(A) + \In_{\ssQ_A}(M/A)
    + (i_\infty, i_\infty, -i_\infty)_{M/A},
  \end{equation}
  where $\In_\ssQ(X)$ stands for the inertia of $X$ restricted to
  the subspace $\ssQ$.

The result now follows by a double application of this formula to
the block Hermitian matrix
\begin{equation*}
	M = \begin{bmatrix}
	A & B \\
	B^* & \Omega
	\end{bmatrix}.
\end{equation*}
Taking the complement with respect to $\Omega$, we find
\begin{align}
	\In(M) &= \In(\Omega) + \In_{\ssQ_\Omega} (M/\Omega) + (i_\infty, i_\infty, -i_\infty)_{M/\Omega}  \nonumber \\
	&= \In(\Omega) + \In(S) + (i_\infty, i_\infty, -i_\infty)_{S}, \label{InM1}
\end{align}
because $(M/\Omega)_{\ssQ_\Omega} = S$ and $i_\infty(M/\Omega) = \rk(B P_\Omega B^*) = i_\infty(S)$. On the other hand, taking the complement with respect to $A$, we find
\begin{align}
	\In(M) &= \In(A) + \In_{\ssQ_A} (M/A) + (i_\infty, i_\infty, -i_\infty)_{M/A} \nonumber \\
	&= \In(A) + \In(W) \label{InM2},
\end{align}
because \eqref{eq:Wregular} implies $P_A B = 0$, hence $\ssQ_A = \Null(B^* P_A B) = \bbC^d$ and 
\[i_\infty(M/A) = \rk(B^* P_A B) = 0. \]
Comparing \eqref{InM1} and \eqref{InM2}, we obtain
\[
  \In(W) = \In(\Omega) + \In(S) - \In(A) + (i_\infty, i_\infty, -i_\infty)_S,
\]
which is precisely \eqref{eq:i-W}--\eqref{eq:i+W}.
  To obtain~\eqref{eq:i+Wbis} from \eqref{eq:i+W} we use
  \begin{equation*}
    i_\infty(S) = N - \dim(\ssQ)
      = \big(i_+(A)+i_-(A)+i_0(A)\big)
      - \big(i_+(S)+i_-(S)+i_0(S)\big).
  \end{equation*}
  This completes the proof.
\end{proof}

\subsection{Proof of  Theorem~\ref{thm:mainI} }\label{sec: proof of main thm}

We are now ready to prove Theorem \ref{thm:mainI}, which for convenience we restate here in an equivalent form.

\begin{theorem}
  \label{E: Theorem i(W)}
  Let $T(\alpha)$ be as in \eqref{eq:Talphadecomp} and $W$ be as defined
  in \eqref{eq:W_def}.  Suppose
  $\lambda(\alpha) = \lambda_n\big(T(\alpha)\big)$ has a critical point at
  $\alpha^\circ$ such that $\lambda(\alpha^\circ)$ is simple and that the
  corresponding eigenvector $f^\circ$ is non-zero on at least one end
  of any crossing edge. 
  
  If $i_-(W)=0$, then
  \begin{equation}
    \label{eq:global_min}
    \lambda(\alpha^\circ) \leq \lambda(\alpha)
    \qquad\mbox{for all }\alpha \in \TT^d,
  \end{equation}
  i.e.  $\lambda(\alpha)$ achieves its
  global minimum at $\alpha^\circ$.
  
  If $i_+(W)=0$, then 
  \begin{equation}
    \label{eq:global_max}
    \lambda(\alpha) \leq \lambda(\alpha^\circ)
    \qquad\mbox{for all }\alpha \in \TT^d.
  \end{equation}
  i.e. $\lambda(\alpha)$ achieves its
  global maximum at $\alpha^\circ$.
\end{theorem}

\begin{proof}
Let
$$A:=T(\alpha^\circ)-\lambda(\alpha^\circ).$$
 Consider first the case $i_-(W) = 0$.
  From Lemma~\ref{lem:Windex}, equation~\eqref{eq:i-W} we get
  \begin{equation*}
    0 = i_-(\Omega) +  i_-(S) + i_\infty(S)
    - i_-\left(A\right),
  \end{equation*}
  and hence, using $i_-\left(A\right) = n-1$,
  \begin{equation*}
    n - i_-(\Omega) - i_\infty(S) = i_-(S) + 1.
  \end{equation*}
  By the definition of negative index, $\lambda_{i_-(S) + 1}(S)$ is the smallest non-negative eigenvalue of $S$, which is 0 by Lemma~\ref{lem:CP_eig0}. Then
  applying Lemma~\ref{lem:combinedWeyl}, we get
  \begin{equation*}
    0 = \lambda_{i_-(S) + 1}(S) = \lambda_{n - i_-(\Omega) - i_\infty(S)}(S)
    \leq \lambda_n\big(T(\alpha)\big) - \lambda(\alpha^\circ),
  \end{equation*}
  completing the proof of inequality~\eqref{eq:global_min}.
  
  For the other case, $i_+(W)=0$, we use Lemma~\ref{rem:Windex}
  and equation~\eqref{eq:i+Wbis}, together with the observation that
  \begin{equation*}
    i_-\left(A\right)
    + i_0\left(A\right) = n,
  \end{equation*}
  because $\lambda(\alpha^\circ)$ is simple, 
  to obtain
  \begin{equation*}
    n + i_+(\Omega) = i_-(S) + i_0(S).
  \end{equation*}
 Now observe that $\lambda_{i_-(S) + i_0(S)}(S)$ is the largest non-positive eigenvalue of $S$, which is 0 by Lemma~\ref{lem:CP_eig0}.
  We then use the upper estimate in Lemma~\ref{lem:combinedWeyl} to obtain
\[
  \lambda_n\big(T(\alpha)\big) - \lambda(\alpha^\circ) \leq \lambda_{n +
    i_+(\Omega)}(S) = \lambda_{i_-(S) + i_0(S)}(S) = 0,
\]
which completes the proof of \eqref{eq:global_max}.
\end{proof}

\section{Real symmetric case: Proof of Theorem \ref{thm:mainII}}
\label{sec:proof_mainII}

From Lemma~\ref{lem:HF} and Theorem~\ref{thm:mainI}, we have the implications
\[
	\text{local minimum at } \alpha^\circ \quad \Longrightarrow  \quad \Re W \geq 0,
\]
and
\[
	W \geq 0 \quad \Longrightarrow  \quad \text{global minimum at } \alpha^\circ,
\]
and similarly for maxima.  We now restrict our attention to the case
of real symmetric $H$, with the goal of relating the spectrum of $W$
to the spectrum of its real part. 
At corner points this is always possible, since $W$ ends up being
real. At interior points, $W$ may be complex. However, for
$d \leq 3$ the real part contains enough information to control the
spectrum of the full matrix. This is no longer true when $d \geq 4$.
These observations are at the heart of Theorem
\ref{thm:mainII}, whose proof we divide into two parts. Section \ref{sec:corner} deals with corner points, while Section \ref{sec:interior} deals with interior ones.

As in the rest of the manuscipt, we fix an arbitrary $1 \leq n
\leq N$ and consider $\lambda_n\big(T(\alpha)\big)$ as a function of $\alpha$,
which we denote by $\lambda(\alpha)$.

\subsection{Corner points: Proof of Theorem \ref{thm:mainII}, case (1)}\label{sec:corner}

The following lemma, combined with Theorem~\ref{thm:mainI} and Lemma~\ref{lem:HF}, yields 
the proof of Theorem \ref{thm:mainII}(1).
\begin{lemma}
  \label{lem:corner_points}
  Assume $T(\alpha)$ is the Floquet--Bloch transform of a real
  symmetric operator $H$.  Let $\alpha^\circ \in \CornerP = \{0,
  \pi\}^d$ and assume that $\lambda(\alpha^\circ)$ is
  simple.  Then, $\alpha^\circ$ is a critical point of
  $\lambda(\alpha)$ and the corresponding matrix $W$ is real.
\end{lemma}

\begin{proof}
  At a corner point $\alpha^\circ$ each $e^{i \alpha^\circ_j}$ is
  real. This means $T(\alpha^\circ)$ is a real symmetric matrix, so we
  can assume that the eigenvector $f^\circ$ is real. It then follows
  from \eqref{eq:Bexplicit} that the matrix $B$ is purely imaginary,
  and hence the vector $B^* f^\circ$ is as well. On the other hand,
  $B^* f^\circ$ is real, since it is the gradient of a real function (by
  Lemma~\ref{lem:HF}), so we conclude that $B^* f^\circ = 0$ and
  hence $\alpha^\circ$ is a critical point.

  We similarly have that $\Omega$ is real (as the Hessian of a real function, see
  \eqref{eqn:BandOmega_def}), $T(\alpha^\circ)-\lambda(\alpha^\circ)$
  is real, and $B$ is imaginary, so we conclude that
  $W = \Omega - B^*\big(T(\alpha^\circ)-\lambda(\alpha^\circ)\big)^+
  B$ is real.
\end{proof}

\begin{remark}
  The condition of $H$ being real can be relaxed.  If the matrix $T_0$
  appearing in the
  decomposition~\eqref{eq:Talphadecomp} is real, then any complex
  phase in the coefficient $h_j$ can be absorbed as a shift of the
  corresponding $\alpha_j$.  Of course, that would shift the location
  of the ``corner points.''

  The condition of real $T_0$ may turn out to hold after a ``change of
  gauge'' transformation.  Combinatorial conditions for the existence
  of a suitable gauge and a suitable choice of the fundamental domain
  were investigated in \cite{HigShi_jfa99,KorSab_jfa17,KorSab_prep18}.
\end{remark}

\begin{remark}
  On lattices whose fundamental domain is a tree, one can also test
  the local character of the extremum at $\alpha^\circ\in\CornerP$ by
  counting the sign changes of the corresponding eigenvector.  More
  precisely, assuming $f^\circ$ is the $n$-th eigenfunction of
  $T(\alpha^\circ)$ and is non-zero on any $v$, the Morse index of
  the critical point $\alpha^\circ \in \CornerP$ was shown in
  \cite{Ber_apde13,Col_apde13} (see also
  \cite[Appendix A.1]{BanBerWey_jmp15}) to be equal to
  $\phi_n - (n-1)$, where
  \begin{equation*}
    \phi_n = \# \{ (u,v) \colon T_{u,v}(\alpha^\circ) f^\circ_u
    f^\circ_v > 0\}.
  \end{equation*}
\end{remark}

\subsection{Interior points: Proof of Theorem \ref{thm:mainII}, cases  (2) and (3).}\label{sec:interior}
Next we deal with the case that $\alpha^\circ \in \TT^d$ is not a
corner point. In this case $W$ is in general complex, so
$\Hess \lambda(\alpha^\circ) = 2 \Re W$ may not contain enough
information to determine the indices $i_\pm(W)$.  However, it turns
out that if $\alpha^\circ \in \TT^d$ is not a corner point, then $0$
must be an eigenvalue of $W$. This provides enough information to
obtain the desired conclusion in dimensions $d=2$ and $3$,
as claimed in  cases (2) and (3) of Theorem \ref{thm:mainII}.

\begin{theorem}
  \label{thm:detW0}
  Assume $T(\alpha)$ is the Floquet--Bloch transform of a real
  symmetric operator $H$ and $\alpha^\circ$ is a critical point of
  $\lambda(\alpha)$, such that $ \lambda(\alpha^\circ)$ is
  simple and the corresponding eigenvector $f^\circ$ is non-zero on at
  least one end of each crossing edge (see Section~\ref{sec:basic defs}).
  Then,  $\alpha^\circ \in \TT^d \setminus \{0,\pi\}^d$ implies
  $i_0(W) \geq 1$.
\end{theorem}

This theorem shows an intriguing contrast between $W$ and the Hessian of $\lambda(\alpha)$, 
the latter of which is the real
part of $W$ and is conjectured to be generically non-degenerate\,---\,see \cite{DoKucSot_jmp20} for a thorough investigation of diatomic
graphs and \cite{FilKac_am18} for a positive result for elliptic operators
on $\RR^2$.

For the proof, we will need the following observation.

\begin{lemma}
  \label{lem:A0_real}
  Under the assumptions of Theorem~\ref{thm:detW0}, the matrix $S$
  defined in \eqref{Sdef} has real entries.
\end{lemma}

\begin{proof}
We recall that the crossing edges for the graph are denoted by $(u_j,v_j)$ 
with $j=1, \dots, d$ (see Section~\ref{sec:basic defs}). We also continue to refer to $J'$ and $J''$ as defined in \eqref{eq:Jprime_def}.
From the decomposition \eqref{eq:Talphadecomp} we have
\[
	T(\alpha^\circ) - \lambda(\alpha^\circ) = T_0  - \lambda(\alpha^\circ) + \sum_{j \in J'} T_j(\alpha_j^\circ) + \sum_{j \in J''} T_j(\alpha_j^\circ),
\]
with $T_0$ and $\lambda(\alpha^\circ)$ real. It was shown in the proof of Lemma \ref{lem:combinedWeyl} that the summands with $j \in J''$ vanish when conjugated by the orthogonal projection $Q$ onto $\ssQ = \Null(BP_{\Null(\Omega)} B^*)$. Hence, it is enough to show that
\[
	\sum_{j \in J'} T_j(\alpha_j^\circ) - B \Omega^+ B^*
\]
is real. Using \eqref{eq:BOB}, we can write this as a sum of terms of the form
\[
	\begin{bmatrix} 0 & h_j e^{i \alpha^\circ_j} \\ \cc{h_j} e^{-i \alpha^\circ_j} & 0 \end{bmatrix} 
	 -   \begin{bmatrix}
    \Omega_{jj} |f^\circ_{u_j}|^{-2} & 
    h_j e^{i \alpha^\circ_j} \\[5pt]
    \cc{h_j} e^{-i \alpha^\circ_j} &
    \Omega_{jj} |f^\circ_{v_j}|^{-2}
  \end{bmatrix}
 = - \begin{bmatrix}
    \Omega_{jj} |f^\circ_{u_j}|^{-2} & 
    0 \\[5pt]
    0 &
    \Omega_{jj} |f^\circ_{v_j}|^{-2}
  \end{bmatrix}
\]
which have real entries by Lemma \ref{lem:Oexplicit}.
\end{proof}

\begin{proof}[Proof of Theorem~\ref{thm:detW0}]
  We first rewrite equation~\eqref{eq:i0W} of Lemma~\ref{rem:Windex} as a
  sum of non-negative terms,
  \begin{equation*}
    i_0(W) = (i_0(\Omega) - i_\infty(S)) + (i_0(S) - 1),
  \end{equation*}
  with $S$ as defined in \eqref{Sdef} and
$i_0(A) = i_0\big(T(\alpha^\circ)- \lambda(\alpha^\circ)\big) = 1$.  The first term is non-negative because $i_0(\Omega) = \rk(P) \geq \rk(B P B^*) = i_\infty(S)$, and the second term is non-negative by Lemma~\ref{lem:CP_eig0}.
  
First, suppose the real and imaginary parts of $f^\circ$ are linearly independent. From Lemma \ref{lem:CP_eig0} we have $f^\circ \in \ssQ$. Because $S$ is real, $\Re f^\circ$ and $\Im f^\circ$ are linearly independent null-vectors of $S$, so we have $i_0(S) \geq 2$ and hence $i_0(W) \geq 1$. Thus, for the remainder of the proof we can assume that the real and imaginary parts of $f^\circ$ are linearly dependent. Multiplying by a complex phase, this is the same as assuming that $f^\circ$ is real.
  
  Since $\alpha^\circ$ is not a corner point, we can assume, without loss of generality, that $\alpha_1 \notin \{0,\pi\}$.
 Using \eqref{cp:real}, the criticality of $\alpha^\circ$ implies that 
  $f^\circ_{u_1} f^\circ_{v_1}\, e^{i \alpha_1^\circ}  \in
  \RR$, and therefore $f^\circ_{u_1} f^\circ_{v_1} = 0$. Since $f^\circ$ is non-zero on at least one 
  end of any crossing edge, we may assume that $f^\circ_{u_1} = 0$ and 
  $f^\circ_{v_1}\neq 0$. From \eqref{eq:Bexplicit} we see that the first column of $B$ has a single nonzero entry, in the $u_1$ component.

From the decomposition \eqref{eq:Talphadecomp} we have $T(\alpha^\circ)_{u_1,v_1} = h_1 e^{i \alpha^\circ_1} \notin \RR$. Considering the $u_1$-th row of the eigenvalue equation $ \lambda(\alpha^\circ)f^\circ = T(\alpha^\circ) f^\circ$, we find
\[
	0 = \lambda(\alpha^\circ) f^\circ_{u_1} = T(\alpha^\circ)_{u_1,v_1} f^\circ_{v_1} + \sum_{v \neq v_1} T(\alpha^\circ)_{u_1,v} f^\circ_{v}.
\]
Since $f^\circ$ is real and $f^\circ_{v_1}\neq 0$, this implies $T(\alpha^\circ)_{u_1,v}$ is non-real for some $v \neq v_1$. This means that there exists another crossing edge, say the $j=2$ edge $(u_2,v_2)$, such that $u_1=u_2$. Then $f^\circ_{u_2} = f^\circ_{u_1} = 0$, so \eqref{eq:Bexplicit} implies that the second column of $B$ is zero except for the $u_1$ component, hence the first and second columns of $B$ are linearly dependent. By Remark \ref{rem:ifB}, this implies $\rk(P) > \rk(B P B^*)$, and hence $i_0(\Omega) - i_\infty(S) \geq 1$, which completes the proof.
\end{proof}

We now discuss what the two conditions, $\Re W \geq 0$ and $\det W =
0$, can tell us about the positivity of the matrix $W$ in dimensions $d
\leq 3$.
In dimension $d=1$ we immediately get $W=0$, hence, by Theorem \ref{E:
  Theorem i(W)}, any non-corner extremum $\lambda(\alpha^\circ)$ is both a
global minimum and a global maximum of $\lambda(\alpha)$.  Therefore,
$\lambda(\alpha)$ is a ``flat band,'' in agreement with the results of
\cite{ExnKucWin_jpa10}.
In dimensions $d=2,3$ we have the following results.

\begin{lemma}
  \label{lem:W2d}
  Let $W$ be a $2\times2$ Hermitian matrix with $\det W = 0$. If
  $\Re W \geq 0$, then $W \geq 0$.
\end{lemma}

\begin{proof}
  If $w$ is the (potentially) non-zero eigenvalue of $W$, we have
  \begin{equation*}
    w = \tr W = \tr \Re W \geq 0,
  \end{equation*}
  and therefore $W \geq 0$.
\end{proof}

\begin{lemma}
  \label{lem:W3d}
  Let $W$ be a $3\times3$ Hermitian matrix with $\det W = 0$. If
  $\Re W > 0$, then $W \geq 0$.
\end{lemma}

\begin{proof}
  For convenience we write $W = A + iB$, where $A$ and $B$ are real
  matrices with $B^T = - B$.  The imaginary part $iB$ is a Hermitian
  matrix with zero trace and determinant. If $B\neq 0$, then
  $i_+(iB)=i_-(iB)=i_0(iB)=1$.  Since $A > 0$, the Weyl inequalities (for
  $W$, as a perturbation of $A$ by $iB$) yield
  \begin{equation*}
    0<\lambda_1(A) \leq \lambda_2(W) \leq \lambda_3(W),  
  \end{equation*}
  forcing $\lambda_1(W)=0$ and therefore $W \geq 0$.
\end{proof}

Theorem \ref{thm:mainII} now follows as a consequence of
Theorems~ \ref{thm:mainI} and \ref{thm:detW0}, and
Lemmas~\ref{lem:HF}, \ref{lem:corner_points}, \ref{lem:W2d} and \ref{lem:W3d}.

\begin{remark}\label{rem:3dpos}
  The strict inequality $\Re W > 0$ in Lemma \ref{lem:W3d} is
  necessary when $d=3$. To see this, consider
  \[
    W = \begin{bmatrix} \epsilon & i & 0 \\
      -i & \epsilon & 0 \\
      0 & 0 & 0 \end{bmatrix}
  \]
for any $\epsilon \in (0,1)$.
The matrix $W$ has eigenvalues $-1 + \epsilon, 0, 1 + \epsilon$, whereas
  $\Re W$ has eigenvalues $0,\epsilon,\epsilon$. That is, $\det W = 0$
  and $\Re W \geq 0$, but $W$ is not non-negative.

  When $d=4$, even strict positivity of $\Re W$ is not enough to
  guarantee $W \geq 0$.  This is illustrated in
  the example of Section~\ref{sec:d4counter} below.
\end{remark}


\section{Examples}
\label{sec:examples}

We present here some illustrative graphs that highlight features of
our results, particularly regarding vanishing
components of the eigenvector and conjectured necessity of the
criterion in Theorem~\ref{thm:mainI}  (Section \ref{sec:eig}).  We also demonstrate that the
restrictions on the number of crossing edges, or, in the case of
Theorem~\ref{thm:mainII}(3), on the dimension $d$, cannot be
dropped without imposing further conditions (Section \ref{sec:assum}).

\subsection{Examples: eigenvectors with vanishing components}\label{sec:eig}

A significant effort in the course of the proofs in
Section~\ref{sec:shift} was devoted to treating eigenvectors with some
zero components.  We were motivated in this effort by some well-known
examples, which we discuss in Sections~\ref{sec:honecomb} and
\ref{sec:lieb}.  In particular, we demonstrate the use of the
generalized Haynsworth formula \eqref{eq:extendedHaynsworth}, needed
here because $\Omega$ is not invertible.  In Section~\ref{example:Kuc}
we revisit the example of \cite{HarKucSob_jpa07} and modify it to test
our conjecture that the condition in Theorem~\ref{thm:mainI} is not
only sufficient but also necessary for the global extremum.  

\subsubsection{Honeycomb lattice}
\label{sec:honecomb}

We consider the honeycomb lattice as shown in
Figure~\ref{fig:more_lattices}(a) whose fundamental domain consists of
two vertices, denoted $\tilde A$ and $\tilde B$.  The tight-binding
model on this lattice was used to study graphite \cite{Wal_pr47} and
graphene \cite{CasGraphen_rmp09,Katsnelson_graphene}.  For some discussions
of the influence of symmetry on the spectrum of this
model, see for instance
\cite{FefWei_jams12,BerCom_jst18}.  We have
\begin{equation}
  \label{eq:Talpha_honeycomb}
  T(\alpha) =
  \begin{bmatrix}
      q_{\tilde A} & -1 - e^{i \alpha_1} - e^{i \alpha_2} \\
      -1 - e^{-i \alpha_1} - e^{-i \alpha_2}  & q_{\tilde B}
    \end{bmatrix},
\end{equation}
where $q_{\tilde A}, q_{\tilde B}$ are the on-site energies for each
sub-lattice.  There is an interior global maximum of the bottom band,
and an interior global minimum of the top band, at
\[
  \alpha^\circ = \left(\frac{2 \pi}{3}, -\frac{2 \pi}{3} \right),
\]
as well as their symmetric copies at $-\alpha^\circ$.  The eigenvalues
are simple unless $q_{\tilde A} = q_{\tilde B}$, in which case the
so-called Dirac conical singularity is formed.

Assume without loss of generality that $q_{\tilde A} < q_{\tilde B}$,
and consider $\lambda = \lambda_1\big(T(\alpha)\big)$.  We have
\begin{equation*}
  \lambda(\alpha^\circ)=q_{\tilde A}, \qquad
  f^\circ =
  \begin{bmatrix}
    1\\0
  \end{bmatrix},  
\end{equation*}
and
\begin{equation*}
  T(\alpha^\circ) - \lambda(\alpha^\circ)
  = 
    \begin{bmatrix}
      0 & 0  \\
      0 &   q_{\tilde B} - q_{\tilde A}
    \end{bmatrix},
  \qquad
  \left(T(\alpha^\circ) - \lambda(\alpha^\circ)\right)^+ 
  = 
    \begin{bmatrix}
      0 & 0  \\
      0 &   (q_{\tilde B} - q_{\tilde A})^{-1}
    \end{bmatrix}.
\end{equation*}
The derivative matrices $B$ and $\Omega$ are 
\begin{equation*}
  B =
  \begin{bmatrix}
    0 & 0  \\
    i e^{-i \frac{2 \pi}{3}} &   i e^{i \frac{2 \pi}{3}}     
  \end{bmatrix}
  \quad\mbox{and}\quad 
  \Omega =
    \begin{bmatrix}
      0 & 0  \\
      0 &   0 
    \end{bmatrix}.
\end{equation*}
As a result,
\begin{equation*}
  W = - \frac{1}{q_{\tilde B} - q_{\tilde A}}
    \begin{bmatrix}
      1 & e^{-i \frac{2 \pi}{3}}   \\
      e^{i \frac{2 \pi}{3}}  &  1
    \end{bmatrix}
\end{equation*}
and $\det(W)=0$ (in agreement with Theorem~\ref{thm:detW0}) and
$W\leq0$ (in agreement with $\lambda(\cdot)$ having the global maximum at
$\alpha^\circ$).
We also observe that
\begin{equation*}
  B P B^* =
    \begin{bmatrix}
      0 & 0  \\
      0 & 2
    \end{bmatrix},
\end{equation*}
giving $\dim \ssQ = 1$,
\begin{equation*}
  S  = 
    \begin{bmatrix}
      1 &0   
    \end{bmatrix}
    \begin{bmatrix}
      0 &0   \\
      0 &  q_{\tilde B} - q_{\tilde A}
    \end{bmatrix}
    \begin{bmatrix}
      1    \\ 
      0    
    \end{bmatrix} 
    = 0.
\end{equation*}
and  $i_\infty (S) = 1$.
  
To illustrate Lemma \ref{lem:Windex}, we now have, with
$A=T(\alpha^\circ) - \lambda(\alpha^\circ)$, 
\begin{align*}
  &1 =  i_- (W) = i_-(\Omega) +  i_- (S) + i_\infty (S) - i_- (A)
    = 0 + 0 + 1 - 0, \\
  &1 =  i_0 (W) = i_0(\Omega) + i_0 (S) - i_\infty (S) - i_0 (A)
    = 2 + 1 -1 - 1, \\
  &0 =  i_+ (W) = i_+(\Omega) - i_-(S) - i_0 (S) + i_- (A) + i_0 (A)
    = 0  - 0 - 1+ 0 + 1.
\end{align*}

We also use this example to demonstrate one of the standard geometric
embeddings of the graph.  Here we follow the conventions of
\cite{CasGraphen_rmp09,FefWei_jams12,BerCom_jst18}.  A slightly
different (but unitarily equivalent) parameterization is traditionally
used in optical lattice studies, see for instance
\cite{haldane1988model,ozawa2019topological} and related references,
though we note here that the latter models often include next-to-nearest
neighbors or further connections which are not covered by our results.

\begin{figure}
  \centering
  \includegraphics{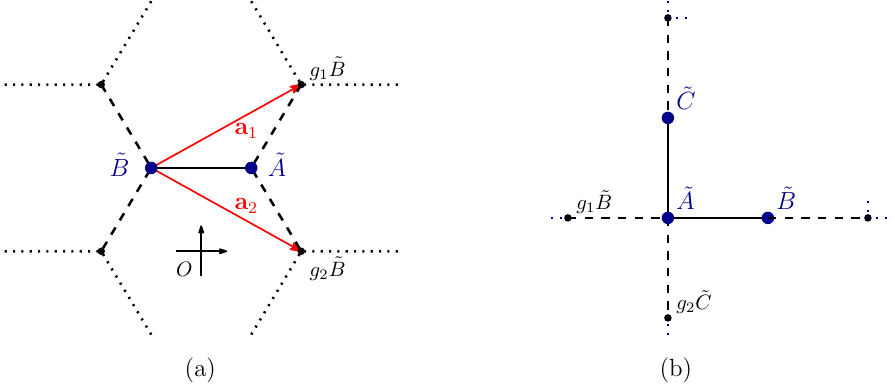}
  \caption{(a) Fundamental domain of the geometric embedding of honeycomb lattice
    resulting in the Floquet--Bloch representation \eqref{HaldaneHam}; (b)
  Fundamental domain of the Lieb lattice.}
  \label{fig:graphene_embedding}
\end{figure}

The triangle Bravais lattice is the set of points $\Lambda = \{n_1 \mathbf{a}_1
+ n_2 \mathbf{a}_2 \colon (n_1,n_2) \in \ZZ^2\}$, where
the vectors
\begin{equation}
  \label{eq:lattice_generators}
  \mathbf{a}_1 =
  \begin{pmatrix}
    \sqrt{3}/2 \\ 1/2
  \end{pmatrix},
  \quad
  \mathbf{a}_2 =
  \begin{pmatrix}
    \sqrt{3}/2 \\ -1/2
  \end{pmatrix}
\end{equation}
represent the periodicity group generators $g_1$ and $g_2$.  Vertices
$\tilde A$ are placed at locations
$\left(\frac1{2\sqrt{3}}, \frac12\right)^T + \Lambda$, while vertices
$\tilde B$ are placed at
$\left(-\frac1{2\sqrt{3}}, \frac12\right)^T + \Lambda$, see
Figure~\ref{fig:graphene_embedding}. This way the geometric graph is
invariant under rotation by $2\pi/3$, while the reflection
$x \mapsto -x$ maps vertices $\tilde A$ to $\tilde B$ and vice versa.

The reciprocal (dual) lattice, $\Lambda^*$, consists of the set of
vectors $\xi$ such that $e^{i v\cdot \xi} = 1$ for every
$v \in \Lambda$. The ``first Brillouin zone'' $\mathcal{B}$, a
particular choice of the fundamental domain in the dual space, is
defined as the Voronoi cell of the origin in the dual lattice.  In
this case it is hexagonal.

The Floquet--Bloch transformed operator parametrized by $\mathbf{k}
\in \mathcal B$ takes the form
\begin{equation}
  \label{HaldaneHam}
  T(\mathbf{k}) =
  \begin{bmatrix}
    q_{\tilde A}
    & -1 - e^{i \mathbf{k} \cdot \mathbf{a}_1}
    - e^{i \mathbf{k} \cdot \mathbf{a}_2} \\
    -1 - e^{-i \mathbf{k} \cdot \mathbf{a}_1}
    - e^{-i \mathbf{k} \cdot \mathbf{a}_2}
    & q_{\tilde B}
  \end{bmatrix}.
\end{equation}
  
While it does not admit a decomposition of the form \eqref{eq:Talphadecomp}, it is
related to $T(\alpha)$ of \eqref{eq:Talpha_honeycomb} by a linear
change of variables and so, by Remark~\ref{rem:brillouin_basis} and
Lemma~\ref{lem:invariance}, we can apply our theorems to the operator
\eqref{HaldaneHam} by directly computing the relevant derivatives in $W$
with respect to the variable $\mathbf{k}$.  We display the dispersion surfaces on the left of Figure \ref{f:dispsurfex}.

\begin{figure}
  \centering
  \includegraphics[width=.4\textwidth]{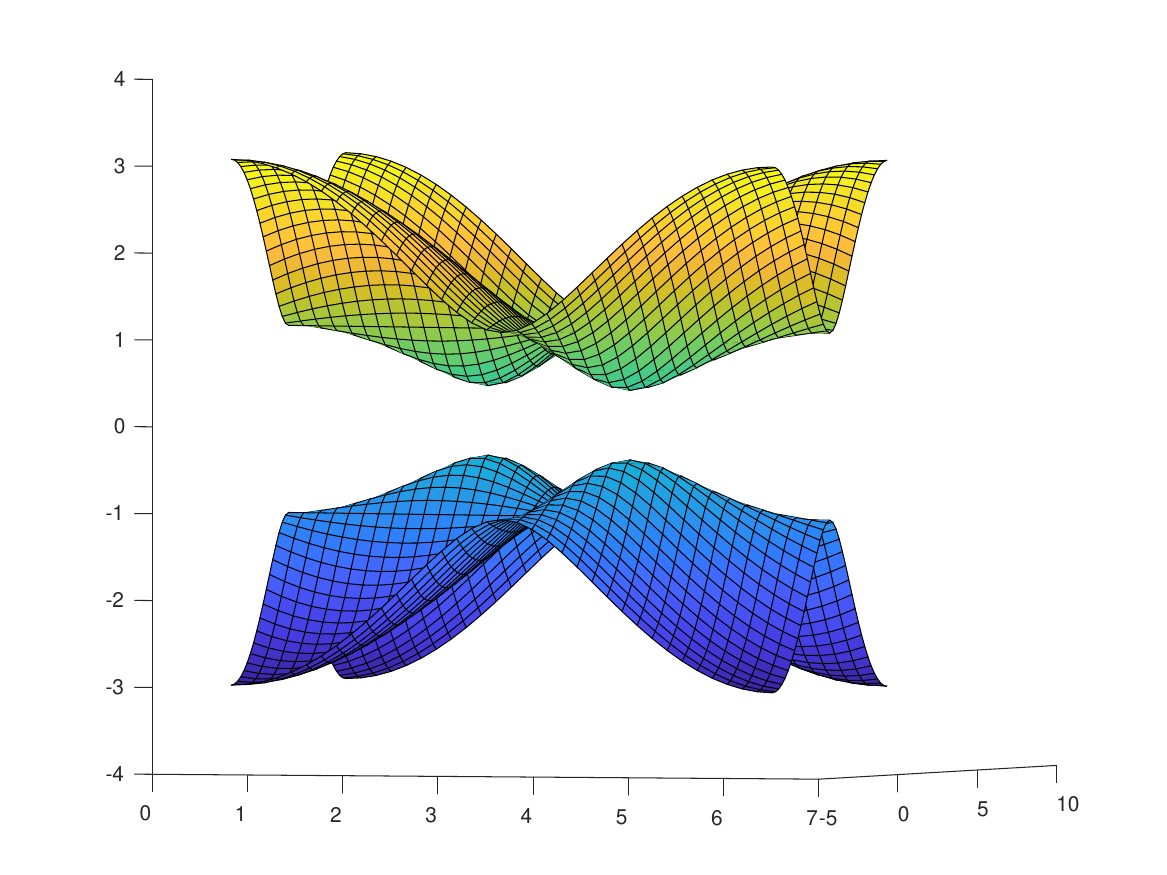}   
  \includegraphics[width=.4\textwidth]{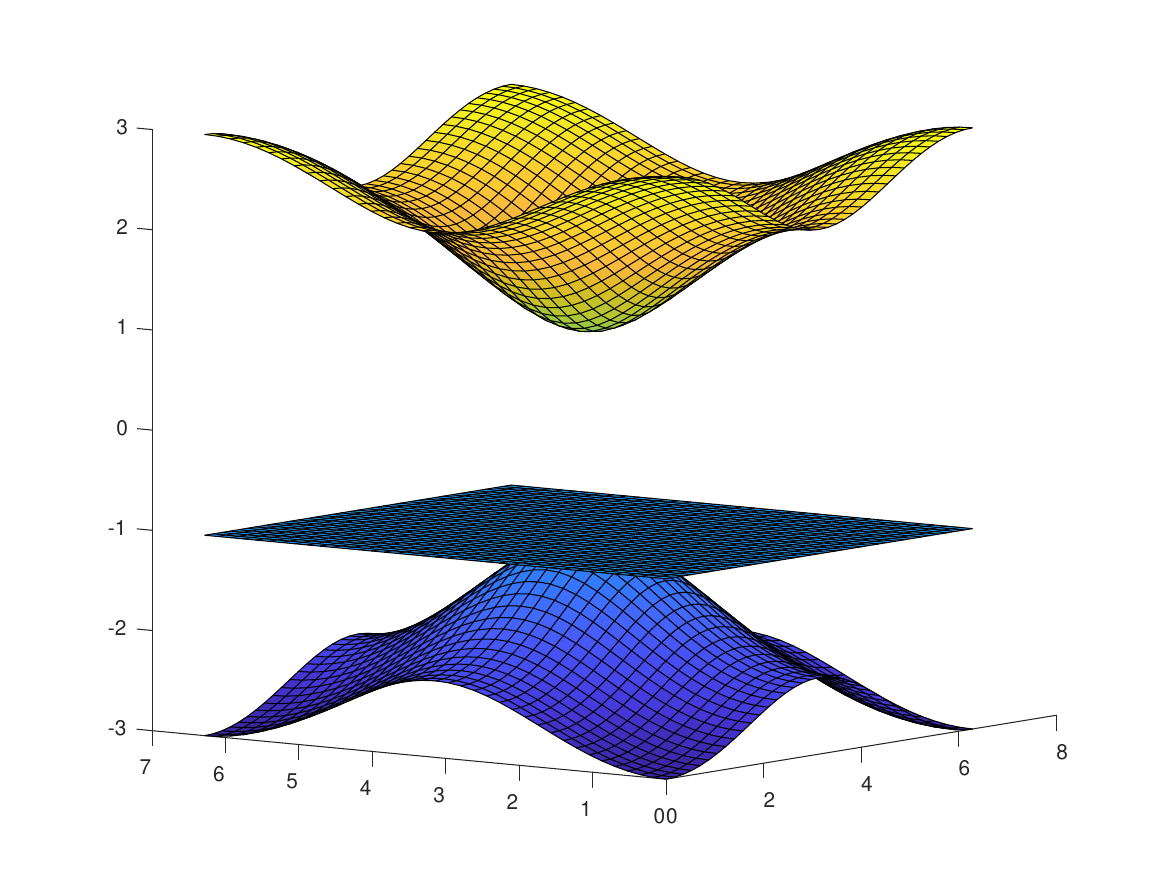} 
  \caption{The dispersion surfaces of \eqref{HaldaneHam}
    (left) and \eqref{Hamiltonian_TrivLieb} (right).}
  \label{f:dispsurfex}
\end{figure}

\subsubsection{Lieb lattice}
\label{sec:lieb}

As another key example of a model that fits into Theorem
\ref{thm:mainII}, we consider a version of the Lieb Lattice graph seen
in Figure \ref{fig:more_lattices}(b), consisting of three copies of
the square lattice as in Figure~\ref{fig:graphene_embedding}(b), with
$q_{\tilde A}, q_{\tilde B}, q_{\tilde C}$ denoting the on-site
energies for each sub-lattice
\cite{guzman2014experimental,marzuola2019bulk,Mukherjee2015,shen2010single}.
The Floquet--Bloch transformed operator is given by
\begin{equation}
   \label{Hamiltonian_TrivLieb}
  T(\alpha) =
  \begin{bmatrix}
    q_{\tilde A} & -1 - e^{i \alpha_1} & -1 - e^{i \alpha_2}  \\
    -1 - e^{-i \alpha_1} & q_{\tilde B} & 0 \\
    -1 - e^{-i \alpha_2}  & 0 & q_{\tilde C}
  \end{bmatrix}.
\end{equation}
Taking $q_{\tilde A} = 1$ and $q_{\tilde B} = q_{\tilde C} = -1$, this has eigenvalues 
\begin{align*}
	\label{onsiteLieb_evs}
	& \lambda_1(\alpha) = - \sqrt{5 + 2 \cos(\alpha_1) + 2 \cos (\alpha_2)},
	\qquad \lambda_2(\alpha) = -1, \\
	&\hspace{1cm} \lambda_3(\alpha) =  \sqrt{5 + 2 \cos(\alpha_1) + 2 \cos (\alpha_2)}.
\end{align*}
We display the dispersion surfaces on the right in Figure \ref{f:dispsurfex}.
In particular, $\lambda_3(\alpha)$ has a minimum at
$\alpha^\circ = (\pi, \pi)$, namely $\lambda_3(\alpha^\circ) = 1$,
with an eigenvector $f^\circ = (1,0,0)^T$ that vanishes on exactly one
end of each crossing edge. We have
\begin{equation*}
  \Omega = 
    \begin{bmatrix}
      0 & 0 \\
      0 & 0 
    \end{bmatrix}
  = \Omega^+,
  \qquad
  B = D \big(T(\alpha) f^\circ \big)\Big|_{\alpha = (\pi,\pi)} =
    \begin{bmatrix}
      0 & 0 \\
      -i & 0 \\
      0 & -i
    \end{bmatrix}
\end{equation*}
and therefore
\begin{equation*}
  T(\pi,\pi) - \lambda^\circ - B\Omega^{+} B^*
  =
    \begin{bmatrix}
      0 & 0  & 0 \\
      0 & -2 & 0 \\
      0 & 0 & -2
    \end{bmatrix},
  \quad
  \text{and} 
  \quad
  B P B^* =
    \begin{bmatrix}
      0 & 0 & 0 \\
      0 & 1 & 0 \\
      0 & 0 & 1 
    \end{bmatrix}, 
\end{equation*}
giving that $\dim \ssQ = 1$.
Then,
\begin{equation*}
  S  =
    \begin{bmatrix}
      1 & 0 & 0   
    \end{bmatrix}
    \begin{bmatrix}
      0 & 0  & 0 \\
      0 & -2 & 0 \\
      0 & 0 & -2
    \end{bmatrix}
    \begin{bmatrix}
      1  \\  0  \\  0
    \end{bmatrix} 
  = 0,
\end{equation*}
and $i_\infty (S) = 2$.
We also compute
\begin{equation*}
  W = \Omega - B^* A^{+} B =
    \begin{bmatrix}
      \frac12 & 0 \\
      0 & \frac12
    \end{bmatrix}.
\end{equation*}

We note that because $\alpha = (\pi,\pi)$ is a corner point,
Theorem~\ref{thm:detW0} ($\det W=0$) does not apply, but
Lemma~\ref{lem:corner_points} ($W$ is real) does.


\subsubsection{A magnetic modification of the example of \cite{HarKucSob_jpa07}}
\label{example:Kuc}

To give an illustration of Theorem~\ref{thm:mainI} with complex $H$, we modify the example in Figure~\ref{fig:periodic_example} by adding a magnetic field.
Consider the Floquet--Bloch operator
\begin{equation}
  \label{e:KucMag}
  T_\beta ( \alpha) =
    \begin{bmatrix}
      0 & 0 & e^{i \alpha_1} & 1 & 1 + i \beta \\
      0 & 0 & 1 & e^{i \alpha_2} & 1 \\
      e^{-i \alpha_1} & 1 & 0 & 1 & 0 \\
      1 & e^{-i \alpha_2} & 1 & 0 & 1 \\
      1 - i \beta & 1 & 0 & 1 & 0 
    \end{bmatrix},
\end{equation}
which, with $\beta=0$, reproduces the example considered in
\cite{HarKucSob_jpa07}.  It was observed in \cite{HarKucSob_jpa07}
that the second dispersion band has two maxima at interior
points, related by the symmetry $\alpha \mapsto -\alpha$ in the Brillouin
zone, see Figure~\ref{fig:KucMag}(left).  Similarly, there are two
internal minima.  Non-zero $\beta$ adds a slight magnetic field
term on the $1 \to 5$ edge of the form and breaks the symmetry in
the dispersion relation.  One maximum becomes larger (and hence the global
maximum) and the other one smaller (merely a local
maximum), as can be seen in Figure \ref{fig:KucMag}(right). 

\begin{figure}
  \centering
  \includegraphics[width=.4\textwidth]{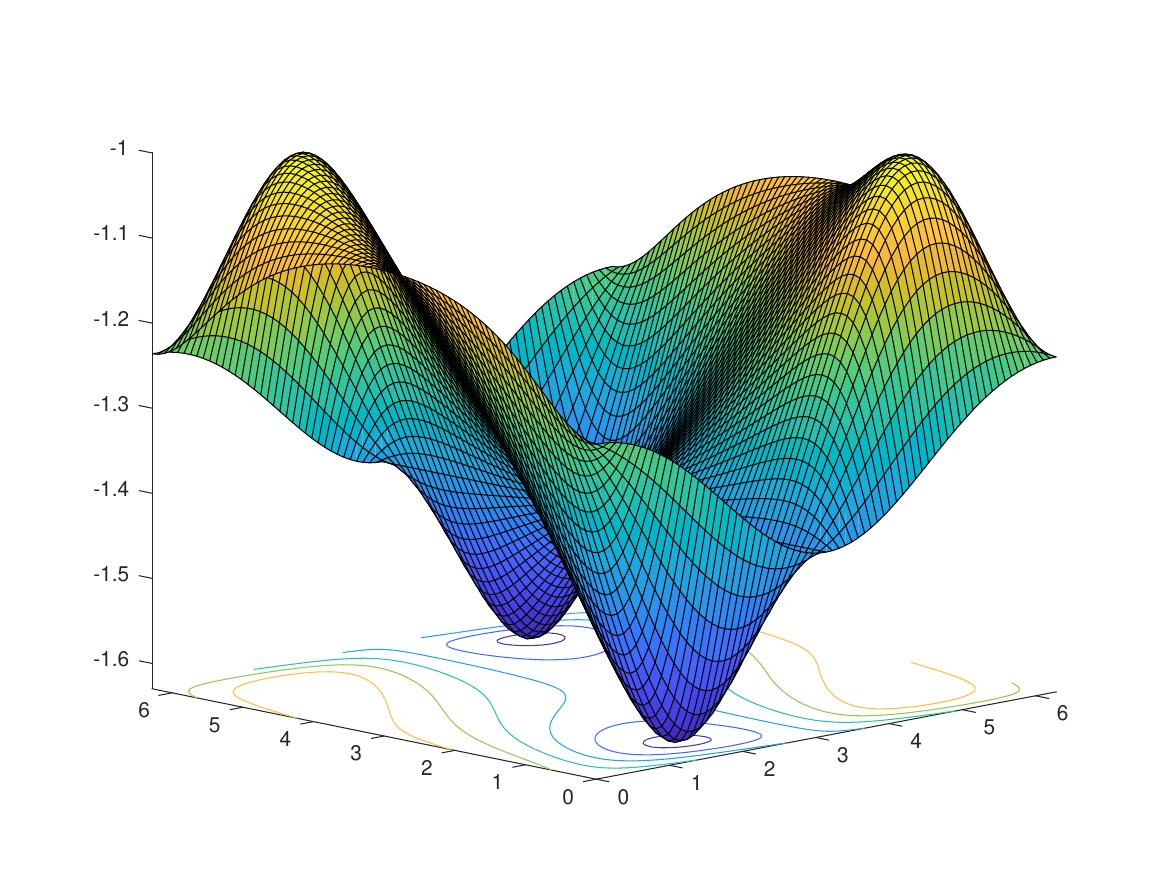}   
  \includegraphics[width=.4\textwidth]{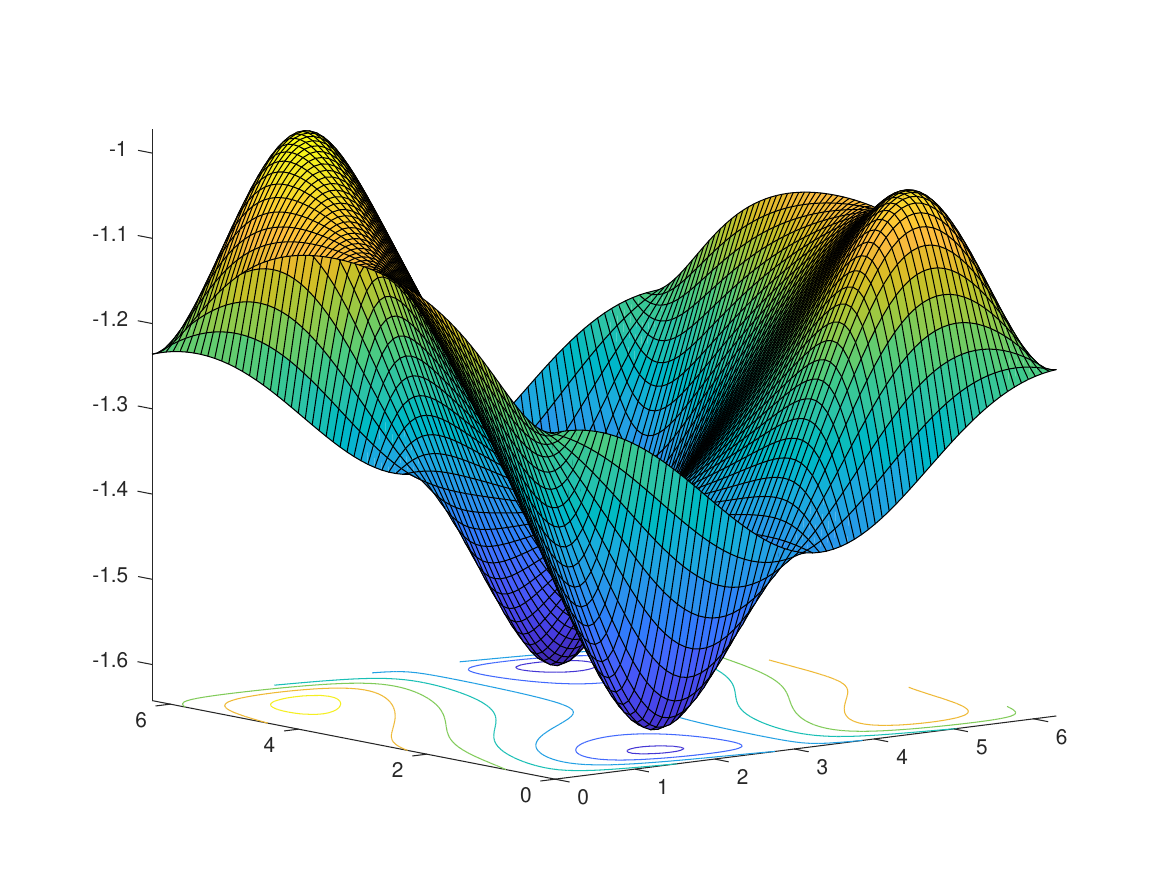} 
  \caption{The second dispersion surface of \eqref{e:KucMag}
    with $\beta=0$ (left) and with $\beta = 0.1$ (right).}
  \label{fig:KucMag}
\end{figure} 

Taking $\beta = 0.1$, the locations of the two maxima of
$\lambda_2\big(T_\beta (\alpha)\big)$ were numerically
computed using {\it Matlab} (both using an optimization solver {\tt
  fminunc} and a root finder {\tt fsolve}) to be at $(\alpha_1^g,
\alpha_2^g) \approx (1.0632,5.2200)$ and $(\alpha_1^\ell, \alpha_2^\ell) \approx
(5.2534,1.0298)$.  Computing their corresponding eigenvectors
$f^\circ$ and using equation~\eqref{Dlambda} the gradient was verified
to be zero with error of less than $4 \times 10^{-16}$ for both
critical points.  
For this model, following Lemmas \ref{lem:Bexplicit} and \ref{lem:Oexplicit} we have 
\begin{equation}
  \Omega = 
  \begin{bmatrix}
- \Re ( e^{i \alpha_1^\circ} f^\circ_3 \overline{f^\circ_1} )& 0 \\
0 & - \Re ( e^{i \alpha_2^\circ} f^\circ_4 \overline{f^\circ_2} )
\end{bmatrix} , \ \  \qquad
 B =  \begin{bmatrix}
i e^{i \alpha_1^\circ} f^\circ_3 & 0 \\
0 & i e^{i \alpha_2^\circ} f^\circ_4 \\
-i e^{-i \alpha_1^\circ} f^\circ_1  & 0 \\
0 & -i e^{-i \alpha_2^\circ} f^\circ_2 \\
0 & 0
\end{bmatrix} 
\end{equation}
and as a result we can easily compute the eigenvalues of
$W = \Omega - B^* (T (\alpha^\circ) - \lambda_2 I)^+ B$.  At the
global maximum, $W$ is found to have two negative eigenvalues,
$\{-0.3433, -0.0095\}$, whereas at the local maximum $W$ is
sign-indefinite with eigenvalues $\{-0.3240, 0.0097\}$, the signs
of which are determined up to errors much larger than those in our
calculations.  Analogous results hold for the global and local minima.

This example motivates the following.
\begin{conjecture}
  \label{conj}
  Under the assumptions of Theorem~\ref{thm:mainI}, a critical point
  $\alpha^\circ$ is a global minimum if and only if $W \geq 0$, and a
  global maximum if and only if $W \leq 0$.
\end{conjecture}

\subsection{(Counter)examples: multiple crossing edges and large dimensions}\label{sec:assum}

In this section we provide examples showing that the assumptions in
our theorems are necessary.  First, in Section
\ref{sec:MEG}, we show that the assumption in
Theorems \ref{thm:mainI} and \ref{thm:mainII} that the graph has one
crossing edge per generator is needed.  Next, in Section \ref{sec:d4counter} we show that,
even when $H$ is real-symmetric, the conclusion of Theorem
\ref{thm:mainII} fails for $d=4$.

The example in Section \ref{sec:MEG} demonstrates one of the simplest possible ways of adding multiple edges per generator in the context of a $2 \times 2$ model $T(\alpha)$, but the form of the operator was motivated by the Haldane model \cite{haldane1988model}, which includes next nearest neighbor complex hopping terms in the form of $T(k)$ given in \eqref{HaldaneHam}. We will observe by directly computing the eigenvalues that the dispersion relation can have a local minimum that is not a global minimum.

\subsubsection{Multiple edges per generator}\label{sec:MEG}

To see that the condition of one edge per generator is required, we first consider a model similar to that of the Honeycomb lattice, but with another edge for one of the generators, specifically given by
\begin{equation*}
T(\alpha) = \begin{bmatrix}
-1 + t \cos(\alpha_2) & -1 - e^{i \alpha_1} - e^{i \alpha_2} \\
 -1 - e^{-i \alpha_1} - e^{-i \alpha_2}  & 1 - t \cos(\alpha_2)
 \end{bmatrix},
 \end{equation*}
 where we have introduced multiple edges per generator and for simplicity chosen $q_{\tilde A} = -1$ and $q_{\tilde B} = 1$.   For $t$ sufficiently large, we observe that the branch for $\lambda_1(\alpha)$ has a local minimum that is not a global minimum, as shown in the dispersion surface plotted in Figure \ref{fig:ce1}, where we have taken $ t = 4$ and thus the lowest dispersion surface is described by the function
 \[
\lambda_1 (\alpha) =  -\sqrt{ 2 ( 6 + \cos (\alpha_1) + \cos(\alpha_1 - \alpha_2) - 3 \cos( \alpha_2) + 4 \cos (2 \alpha_2))}.
 \]
The local minimum here occurs at $\alpha = (0,0)$, which is a corner point, and hence we have that $W = \Re W$ is non-negative.  Therefore, this gives a counterexample to both Theorems ~\ref{thm:mainI} and~\ref{thm:mainII} in the case of multiple edges per generator. 

 \begin{figure}
  \centering
  \includegraphics[width=.4\textwidth]{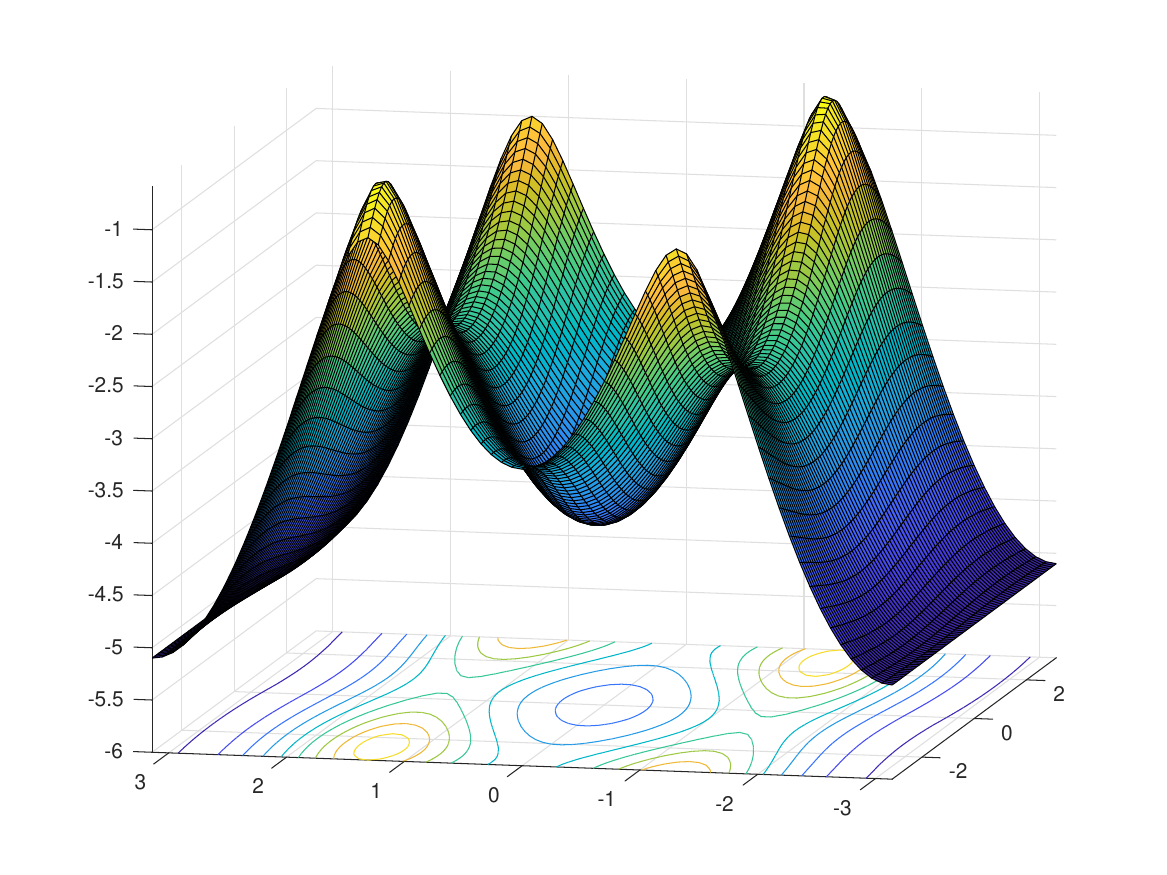}
  \caption{An example of a dispersion band on a $\mathbb{Z}^2$-periodic graph $\Gamma$ with a local (but not global) minimum resulting from multiple edges per generator.}
  \label{fig:ce1}
\end{figure}

This example was motivated by the Haldane model, which is
$\ZZ^2$-periodic. However,  even $\ZZ^1$-periodic graph operators are not
immune from this problem, see \cite{ExnKucWin_jpa10} and \cite[Example
1]{Shi_cmp14}.

\subsubsection{Dimension $d \geq 4$}
\label{sec:d4counter}

We construct here a $\ZZ^4$-periodic graph that displays a
local extremum that is not a global extremum.  The example was found
by searching through positive rank-one perturbations of a random
symmetric matrix having $1$ as a degenerate eigenvalue; this ensured that $1$
is a local (but not necesarily global) maximum.  As a trigger for
terminating the search we used the conjecture in Section~\ref{conj}: the matrix $W$
was computed and the search was stopped when it was sign-indefinite.
The resulting example reveals the presence of a global maximum
elsewhere, thus also serving as a numerical confirmation of the
conjecture's veracity.  We report it with all entries rounded
off for compactness.

\vspace{-5pt}
\begin{tiny}
  \begin{align*}
    \label{d4counter}
    & T( \alpha)  = \\
    & \left[
      \begin{array}{rrrrr}
        2.556782 & .104696 & -.000742 & -.049562 & -.072260 \\
        .104696  & 3.69455 & -.436154 & -.126495 & -.571811 \\
        -.000742 &-.4361543 & 15.033535 & .139015 & -.363838 \\
        -.049562 & -.126495  &  .139015 &  2.146425 & .298246  \\
        -.072260   &  -.571811   & -.363838 &  .298246 &  9.097398
      \end{array} \right] + \left[ \begin{array}{ccccc}
  0 & e^{i \alpha_1} &e^{i \alpha_2}&-e^{i \alpha_3} & e^{i \alpha_4} \\
  e^{-i \alpha_1}  & 0 &0 &0  &0  \\
  e^{-i \alpha_2}  &0  & 0 & 0 & 0  \\
    -e^{-i \alpha_3} & 0 & 0 &  0 & 0   \\
  e^{-i \alpha_4}   & 0  & 0 & 0 & 0 
    \end{array} \right].
  \end{align*}
\end{tiny}

Using the objective function of the form
$\lambda_1\big(T(\alpha)\big)$ and running a Newton BFGS optimization
with randomly seeded values of $\alpha$, we find two distinct local
maxima at $ \alpha^\circ \approx (-1.488, -2.153, 1.553, -3.324)$ and
$\lambda_1(\alpha^\circ) \approx 0.989459$ (close but not equal to $1$ due
to rounding off the entries of the example matrix).  However, the
observed global maximum is $\lambda_1 (\pi,0,\pi,0) \approx 1.2467$.  Hence,
we observe that the corner point is a local maximum that is in fact a
global maximum (as follows from Theorem~\ref{thm:mainII}(1)), but the interior
point is a local maximum that is not a global maximum. The minimum of
the second band appears to be $2.63496$, hence there are no
degeneracies arising between the first two spectral bands.

\appendix
\section{A generalized Haynsworth formula}
\label{sec:gen_Haynsworth}

The inertia of a Hermitian matrix $M$ is defined
to be the triple
\begin{equation}
  \label{eq:inertia}
  \In(M) = (i_+(M), i_-(M), i_0(M))
\end{equation}
of numbers of positive, negative and zero eigenvalues of $M$, respectively. For a Hermitian matrix in block form,
\begin{equation}
  \label{eq:block_form1}
  M =
  \begin{bmatrix}
    A & B \\
    B^* & C
  \end{bmatrix},
\end{equation}
the Haynsworth formula \cite{Hay_laa68} shows that, if $A$ is invertible, then
\begin{equation}
  \label{eq:Haynsworth1}
  \In(M) = \In(A) + \In(M/A),
\end{equation}
where 
\begin{equation}
  \label{eq:Schur_complement}
  M/A := C - B^* A^{-1} B
\end{equation}
is the Schur complement of the block $A$.  We are concerned with the
case when the matrix $A$ is singular.  In this case, inequalities
extending \eqref{eq:Haynsworth1} have been obtained by Carlson et
al. \cite{CarHayMar_siamjam74} and a complete formula was derived by
Maddocks \cite[Thm 6.1]{Mad_laa88}.  Here we propose a different
variant of Maddocks' formula.  Our variant makes the correction terms
more transparent and easier to calculate; they are motivated by a
spectral flow picture.  They are also curiously similar to the answers
obtained in a related question by Morse \cite{Mor_rrmpa71} and Cottle
\cite{Cot_laa74}.

\begin{theorem} 
\label{thm:generalized_Haynsworth}
Suppose $M$ is a Hermitian matrix in the block form \eqref{eq:block_form1}, and let $P$ denote the orthogonal projection onto $\Null(A)$. Then
\begin{equation}
	\label{eq:extendedHaynsworth1}
	\In(M) = \In(A) + \In_{\ssQ}(M/A) + (i_\infty, i_\infty, -i_\infty),
\end{equation}
where  the subspace $\ssQ$ is defined by
\begin{equation}
\label{eq:subspace_Q}
	\ssQ = \Null(B^*P B), 
\end{equation}
$\In_\ssQ(X)$ stands for the inertia of $X$ restricted to the subspace $\ssQ$, and $i_\infty$ is given by
\begin{equation}
\label{eq:i_infinity}
	i_\infty = i_\infty(M/A) = \rk(B^*PB) 
	= \dim(C) - \dim(\ssQ).
\end{equation}
\end{theorem}

\begin{remark}
  \label{rem:any_complement}
  If the matrix $A$ is singular, equation~\eqref{eq:Schur_complement}
  is not appropriate for defining the Schur complement.  It is usual
  to consider the generalized Schur complement 
  \[
  M/A := C - B^* A^+ B,
  \]
  where $A^+$ is the Moore--Penrose pseudoinverse, which is what we have done in the main arguments above.  However, because of
  the restriction to $\ssQ$, any reasonable generalization will work in
  equation~\eqref{eq:extendedHaynsworth1}.  For example, 
  \begin{equation}
    \label{eq:epsilon_Schur}
    M/A_\epsilon := C - B^* (A+\epsilon P)^{-1} B,
  \end{equation}
is well defined for any $\epsilon \neq 0$. Taking the limit $\epsilon \to \infty$, we recover the definition
  with $A^+$. 
In fact, it can be shown that
  \begin{equation*}
    M/A_\epsilon = M/A
    - \tfrac1\epsilon B^*P B,
  \end{equation*}
with the last summand being identically zero on the subspace $\ssQ$.  It follows that the restriction $(M / A_\epsilon)_\ssQ = (M/A)_\ssQ$ is \emph{independent of $\epsilon$}, so the index $\In_\ssQ(M/A_\epsilon)$ is as well.
\end{remark}

\begin{remark}
  The index $i_\infty(M/A)$ has a beautiful geometrical meaning: it is
  the number of eigenvalues of $M/A_\epsilon$ which escape to infinity
  as $\epsilon \to 0$.  Correspondingly, $\In_Q(M/A)$ counts the
  eigenvalues of $M/A_\epsilon$ converging to positive, negative and zero \emph{finite} limits as $\epsilon\to 0$.
\end{remark}

\begin{remark}
  As a self-adjoint linear relation, the Schur complement $M/A$ is
  well defined even if $A$ is singular (see, for example, \cite{CdV_aif99}).
  Then the index $i_\infty(M/A)$ has the meaning of the dimension of
  the multivalued part whereas $\In_Q(M/A)$ is the inertia of
  the operator part of the linear relation (see, for
  example, \cite[Sec 14.1]{Schmudgen_UnboundedSAO} for relevant
  definitions).
\end{remark}

The proof of Theorem~\ref{thm:generalized_Haynsworth} follows simply
from the following formula, which was proved in the generality we
require in \cite{Jon+_laa87} (inspired by a reduced version appearing
in \cite{HanFuj_laa85}).  The original proofs are of ``linear
algebra'' type.  For geometric intuition we will provide a ``spectral
flow'' argument in Section~\ref{sec:HanFujiwara_proof}.

\begin{lemma}[Jongen-M\"obert-R\"uckmann-Tammer, Han-Fujiwara]
  \label{lem:HanFujiwara}
  The inertia of the Hermitian matrix
  \begin{equation}
    \label{eq:A0}
    M =
    \begin{bmatrix}
      0_m & B \\
      B^* & C
    \end{bmatrix},
  \end{equation}
  where $0_m$ is the $m\times m$ zero matrix, is given by the formula
  \begin{equation}
    \label{eq:HanFujiwara}
    \In(M) = \In_{\Null(B)}(C) + (\rk(B), \rk(B), m-\rk(B)). 
  \end{equation}
\end{lemma}

\begin{proof}[Proof of Theorem~\ref{thm:generalized_Haynsworth}]
  Take $A$ and $M$ as given by \eqref{eq:block_form1}.   Let $V = (V_1\ V_0)$ be the unitary matrix of eigenvectors of $A$,
  with $V_0$ being the $m=\dim\Null(A)$ eigenvectors of eigenvalue
  $0$.  We have
  \begin{equation*}
    V^* A V =
    \begin{bmatrix}
      \Theta & 0 \\
      0 & 0_m
    \end{bmatrix},
  \end{equation*}
  where $\Theta$ is the non-zero eigenvalue matrix of $A$ and only the
  most important block size is indicated.  We recall that, with the above
  notation, the Moore--Penrose pseudoinverse is given by
  $A^+ = V_1 \Theta^{-1} V_1^*$.

  Conjugating $M$  by the block-diagonal matrix
  $\diag(V, I)$, we obtain the unitary equivalence
  \begin{equation*}
    M \simeq
    \begin{bmatrix}
      \Theta & 0 & V_1^*B \\
      0 & 0 & V_0^*B \\
      B^*V_1 & B^* V_0 & C
    \end{bmatrix}.
  \end{equation*}
  Applying the Haynsworth formula to the invertible matrix $\Theta$, we get
  \begin{equation*}
    \In(M) = \In(\Theta) + \In
    \begin{bmatrix}
      0 & V_0^*B \\
      B^* V_0 & C - B^*V_1 \Theta^{-1} V_1^* B
    \end{bmatrix}.
  \end{equation*}
We now apply Lemma~\ref{lem:HanFujiwara} to get
	\begin{equation*}
	\In(M) = \In(\Theta)
	+ \In_Q\left(C - B^*V_1 \Theta^{-1} V_1^* B\right)
	+ (i_\infty, i_\infty,m  -i_\infty),
\end{equation*}
since $\Null(V_0^* B) = \Null(B^* PB) = \ssQ$
and $\rk(V_0^*B) = \rk(B^*PB)  = i _\infty$.  We finish the 
proof by observing that $\In(\Theta) + (0,0,m) = \In(A)$ and $C - B^*V_1 \Theta^{-1} V_1^* B$ is equal to the generalized Schur complement $C-B^*A^+B = M/A$.
\end{proof}

\subsection{An alternative proof of Lemma~\ref{lem:HanFujiwara}}
\label{sec:HanFujiwara_proof}

To give a perturbation theory intuition behind
Lemma~\ref{lem:HanFujiwara}, define
\begin{equation}
  \label{eq:Aepsilon}
  M_\epsilon =
  \begin{bmatrix}
    \epsilon I_m & B \\
    B^* & C
  \end{bmatrix}.
\end{equation}
For $\epsilon > 0$, $M_\epsilon$ is a non-negative perturbation of
$M$.  When $\epsilon$ is small enough, none of the negative
eigenvalues of $M$ will cross $0$, therefore $i_-(M_\epsilon) =
i_-(M)$.  Applying the Haynsworth formula to the invertible matrix $\epsilon I$,
we get
\begin{equation*}
  i_-(M) = i_-(M_\epsilon)
  = i_-(\epsilon I) + i_-\left(C - \tfrac1\epsilon B^*B\right)
  = i_-\left(C - \tfrac1\epsilon B^*B\right).
\end{equation*}
Due to the presence of $\frac1\epsilon$, some eigenvalue of
$M/\epsilon := C - \frac1\epsilon B^*B$ becomes unbounded.  More precisely, the
Hilbert space on which $C$ is acting can be decomposed as
\begin{equation}
  \label{eq:decomp}
  H_C = \Ran(B^*B) \oplus \Null(B^*B).  
\end{equation}
There are $\rk (B^*B)$
eigenvalues of $M/\epsilon$ going to $-\infty$ as
$\epsilon\to0$.  The rest of the eigenvalues of
$M/\epsilon$ converge to eigenvalues of $C$ restricted to
$\Null(B^*B)$.  Informally, the operator $M/\epsilon$ is
reduced by the above Hilbert space decomposition in the limit
$\epsilon\to0$.  This argument can be made precise by applying
the Haynsworth formula to $M/\epsilon$ written out in the block
form in the decomposition \eqref{eq:decomp}.

The negative eigenvalues of $i_-(M_\epsilon)$ thus come from
$\rk(B^*B)=\rk(B)$ eigenvalues going to $-\infty$, and the negative
eigenvalues of $C$ on $\Null(B^*B) = \Null(B)$.  This establishes the
negative index in equation~\eqref{eq:HanFujiwara}.  Positive
eigenvalues are calculated similarly by considering small negative
$\epsilon$, and the zero index can be obtained from the total dimension.

\bibliographystyle{myalpha}
\bibliography{bk_bibl}

\end{document}